\documentclass{statsoc}
\usepackage{times}
\usepackage{fullpage}
\usepackage{color,graphicx}
\usepackage{eurosym}
\usepackage{amsmath}
\usepackage{amssymb}
\usepackage{verbatim}
\usepackage{mathrsfs}

\usepackage{footmisc}
\usepackage[utf8]{inputenc}

\usepackage{url}
\usepackage{algorithm}
\usepackage{algorithmicx}
\usepackage{algpseudocode}
\newtheorem{defn}{Definition}
    
\newcommand{\bs}{\mathbf{s}}

\newcommand{\cS}{{\mathcal S}}

\newcommand{\cT}{{\mathcal T}}

\newcommand{\Su}{\mathcal{S}^{\cup}}

\newcommand{\muu}{\mu^{\cup}}

\newcommand{\dif}{\mathrm{d}}

\definecolor{mypine}{rgb}{0.05,0.45,0.05}
\definecolor{mydarkred}{rgb}{0.75,0.05,0.05}
\definecolor{mypurple}{rgb}{0.75,0.0,0.75}
\definecolor{mygray}{rgb}{0.35,0.35,0.35}
\definecolor{mylightgray}{rgb}{0.5,0.5,0.5}

\def\tt{\tilde{t}}

\def\tF{\tilde{F}}
\def\tG{\tilde{G}}

\def\hlambda{\hat{\lambda}}

\def\matern{Mat\'{e}rn }

\newcommand{\given}{\,|\,}

\newtheorem{thrm}[defn]{Theorem} 
\newtheorem{coro}[defn]{Corollary}

\title{Bayesian inference for \matern repulsive processes}
\author{Vinayak Rao\thanks{Corresponding author}}
\address{Department of Statistics, Purdue University, USA}
\author{Ryan P. Adams}
\address{School of Engineering and Applied Sciences, Harvard University, USA}
\author{David Dunson}
\address{Department of Statistical Science, Duke University, USA}
\date{}

\begin{document}

\maketitle
\begin{abstract}
In many applications involving point pattern data, the Poisson process assumption is unrealistic, with the data exhibiting a
more regular spread.
Such repulsion between events is exhibited by trees for example, because of competition for light and nutrients. Other examples include the locations of 
biological cells and cities, and the times of neuronal spikes. Given the many 
applications of repulsive point processes, there is a surprisingly limited literature developing flexible, realistic and interpretable models, as well as 
efficient inferential 
methods.  We address this gap by developing a modeling framework around the Mat\'ern type-III repulsive process. 
We consider a number of extensions of the original \matern type-III process for both the homogeneous and inhomogeneous cases.
We also derive the probability density of this generalized \matern process, allowing us to characterize the conditional distribution of the various 
latent variables, and leading to a novel and efficient Markov chain Monte Carlo algorithm.
We apply our ideas to datasets involving the spatial locations of trees, nerve
fiber cells and Greyhound bus stations.
\end{abstract}

\section{Introduction}

Point processes find wide use in fields such as 
astronomy \citep{Peebles74}, biology \citep{WallSar11}, ecology \citep{Hill1973}, epidemiology \citep{Knox04}, geography \citep{kendall39} and 
neuroscience \citep{Brown2004a}. The simplest and most popular model for point pattern data is the Poisson process (see for example \citet{DalVer2008a}); however, implicit in the 
Poisson process is the assumption of independence of the event locations. This simplification is unsuitable for many
real applications in which it is of interest to account for interactions between nearby events. 

Point processes on the real line can deviate from Poisson either by being more bursty or more refractory. In higher dimensions, 
these are called {clustered} and {repulsive} processes.  Our focus is on the latter,
characterized by being more regular (underdispersed) than the Poisson process. Reasons for this could be competition for finite resources 
between trees \citep{Strand72}, interaction between rigid objects such as cells \citep{WallSar11} or 
the result of planning (for example, the locations of bus stations).
Characterizing this repulsion is important towards understanding how disease affects cell locations~\citep{WallSar11}, how neural spiking
history affects stimulus response~\citep{Brown2004a} or how reliable a network of stations is~\citep{Sahin2007}.


Developing a flexible and tractable statistical framework to study such 
repulsion is not straightforward on spaces more complicated than the real line.  
For the latter, the ordering of points leads to a convenient framework based on renewal processes \citep{DalVer2008a}.
This work focuses on a class of repulsive point processes on higher-dimensional spaces called the \matern type-III process. First introduced in \citet{Matern60, Matern86}, 
this involves thinning events of a `primary' Poisson process that are `too close to each other'. 
Starting with the simplest such process (called a hardcore process), we introduce various extensions that provide a flexible framework for modeling
repulsive processes. We derive the probability density of the resulting process, a characterization that allows us to identify the conditional 
distribution over the thinned events as a simple Poisson process. This allows us to develop a simple and efficient Markov chain Monte Carlo algorithm for 
posterior inference. 

\section{Related work}

{Gibbs processes} \citep{DalVer2008a} arose from the statistical physics literature to describe 
systems of interacting particles. A Gibbs process assigns a potential energy $U(S)$ to any configuration of events $S = (s_1, \cdots,s_n )$: 
\vspace{-.1in}
\begin{align}
  U(S) &= \sum_{i=1}^n \sum_{1 \le j_1 < \cdots < j_i \le n} \psi_i(s_{j_1},\cdots,s_{j_i}),
\end{align}
where $\psi_i$ is an $i$th-order potential term.
Usually, interactions are limited to be pairwise, 
and by choosing these potentials appropriately, one can model different kinds of interactions. 
The probability density of any configuration is proportional to its exponentiated
negative energy, and letting $\theta$ parametrize the potential energy, we have
\begin{align}
  p(S\given\theta) &= \frac{\exp(-U(S;\theta))}{Z(\theta)}.
\end{align}
Unfortunately, evaluating the normalization constant $Z(\theta)$ is usually intractable, making
even sampling from the prior difficult (typically, this requires a coupling from the past approach \citep{Moller2007}). 
Inference over the parameters usually proceeds by maximum likelihood or pseudolikelihood methods \citep{Moller2007, Mateu2001}, and is slow and expensive.

Determinantal point processes \citep{hough_dpp,Scardicchio09}) are another framework for modeling repulsion. While
these processes are mathematically and computationally elegant, they are not intuitive or easy to specify, and with few exceptions (e.g., \citet{affandi2014learning}) there is little Bayesian work involving such models. 

There has also been work using \matern type-I and II processes (see \cite{WallSar11}). We mention some limitations of these next. Additionally,
our sampling scheme outlined later does not extend to these processes, making Bayesian inference via MCMC less natural than with the
type-III process.

\section{\matern repulsive point processes}

\cite{Matern60} proposed three schemes, now called the \matern type-I, type-II, and type-III hardcore point processes,
for constructing repulsive point processes.
For the type-I process, one samples a \emph{primary} process from a homogeneous Poisson process 
with intensity $\lambda$, and then deletes all points separated by less than $R$. While the simplicity of this scheme makes it amenable to theoretical 
analysis, the thinning strategy is often too aggressive. In particular, 
as $\lambda$ increases the thinning effect begins to dominate, and eventually the density of \matern events begins to {decrease} with $\lambda$.

The \matern type-II process tries to rectify this by assigning each point an `age'.
When there is a conflict between two points, 
rather than deleting {both}, the older point wins. While this construction allows higher event densities, it also allows 
an event to be thinned by an earlier point that was also thinned. This is slightly unnatural, as one might expect only surviving points to influence 
future events. 

The \matern type-III process 
addresses these limitations by letting a newer event be thinned only if 
it falls within radius $R$ of 
an older event that was not thinned before. 
The resulting process has a number of desirable properties. In many applications, this thinning mechanism is 
more natural than the type-I and II processes. 
In particular, it forms a realistic model for various spatio-temporal phenomena where the latent birth times are not observed, and
must be inferred.  This process also supports higher event densities than the type-I and 
type-II processes with the same parameters; in fact, as $\lambda$ increases,
the average number of points in any area increases monotonically to 
the `jamming limit' (the maximum density at which spheres of radius $R$ can be packed in a bounded area \citep{moller10}). 
The monotonicity property with respect to the intensity of the primary process is important in applications where we model inhomogeneity by
allowing $\lambda$ to vary with location (see section~\ref{sec:inhom_mat}), with large values of $\lambda$ implying high densities.

In spite of these properties, the \matern type-III process has not found widespread use in the spatial point process community.
Theoretically it is not as well understood as the other two \matern processes; for instance, there is no closed form 
expression for the average number of points in any region. Instead, one usually resorts to simulation studies to better understand the modeling assumptions
implicit to this process.

A more severe impediment to the use of this process (and the other \matern processes) is that given a realization,
there do not exist efficient techniques for inference over parameters 
such as $\lambda$ or the radius of interaction $R$. The few existing inference schemes involve imputing, and then perturbing the 
thinned events via 
incremental birth-death steps. This sets up a Markov chain which proceeds by randomly inserting, deleting or shifting the thinned events, with
the various event probabilities set up so that the Markov chain converges to the correct posterior over thinned events \citep{moller10, Hube:Wolp:2009, adamsthesis}. 
Given the entire set of thinned events, it is straightforward to obtain samples of
the parameters~$\lambda$ and~$R$. However, the incremental nature of these 
birth-death updates can make the sampler mix quite slowly. The birth-death sampler can be adapted to a coupling from
the past scheme to draw perfect samples of the thinned events \citep{Hube:Wolp:2009}. This can then be used to approximate the likelihood of the \matern observations, or perhaps, to 
drive a Markov chain following ideas from \cite{AndRob10}. However, this too can be quite inefficient, 
with long waiting times until the sampler returns a perfect sample.

Somewhat surprisingly, despite being more complicated than the type-I and II processes, 
we can develop an efficient MCMC sampler for the type-III process.
Before describing this, we develop more general extensions of the \matern type-III process, providing a flexible and practical 
framework for the modeling of repulsive processes.

\section{Generalized \matern type-III processes} \label{sec:gen_mat}

   A \matern type-III hardcore point process on a measurable space~$(\cS,\Sigma)$
 is a repulsive point process parametrized by an intensity~$\lambda$ and an 
interaction radius~$R$. It is obtained by thinning events
of a homogeneous {primary} Poisson process~$F$ with intensity~$\lambda$. 
Each event~${f \in F}$ of the primary process is independently assigned a random mark~$t$, the time of its
birth. Without loss of generality, we assume this takes values in the interval~$[0,1]$, which
we call~$\cT$. 
The~$(f_i, t_i)$ form a Poisson process~${F}^{+}$
on~${\cS \times \cT}$ (whose intensity is still~$\lambda$). 
Call the collection of birth times~$T^F$, and define~${F^{+} \equiv (F,T^F)}$ as the collection of (location,birth time) pairs. $T^F$ induces an ordering on the 
events in~${F}^{+}$, 
and a secondary process~${G^{+} \equiv (G, T^G)}$ is obtained by traversing~$F^{+}$ in this order and deleting all points within 
a distance~$R$ of any {earlier}, {undeleted} point. We obtain the \matern process~$G$ by projecting~$G^{+}$ onto~$\cS$.

   \begin{figure}
   \centering
     \includegraphics[width=.3\textwidth]{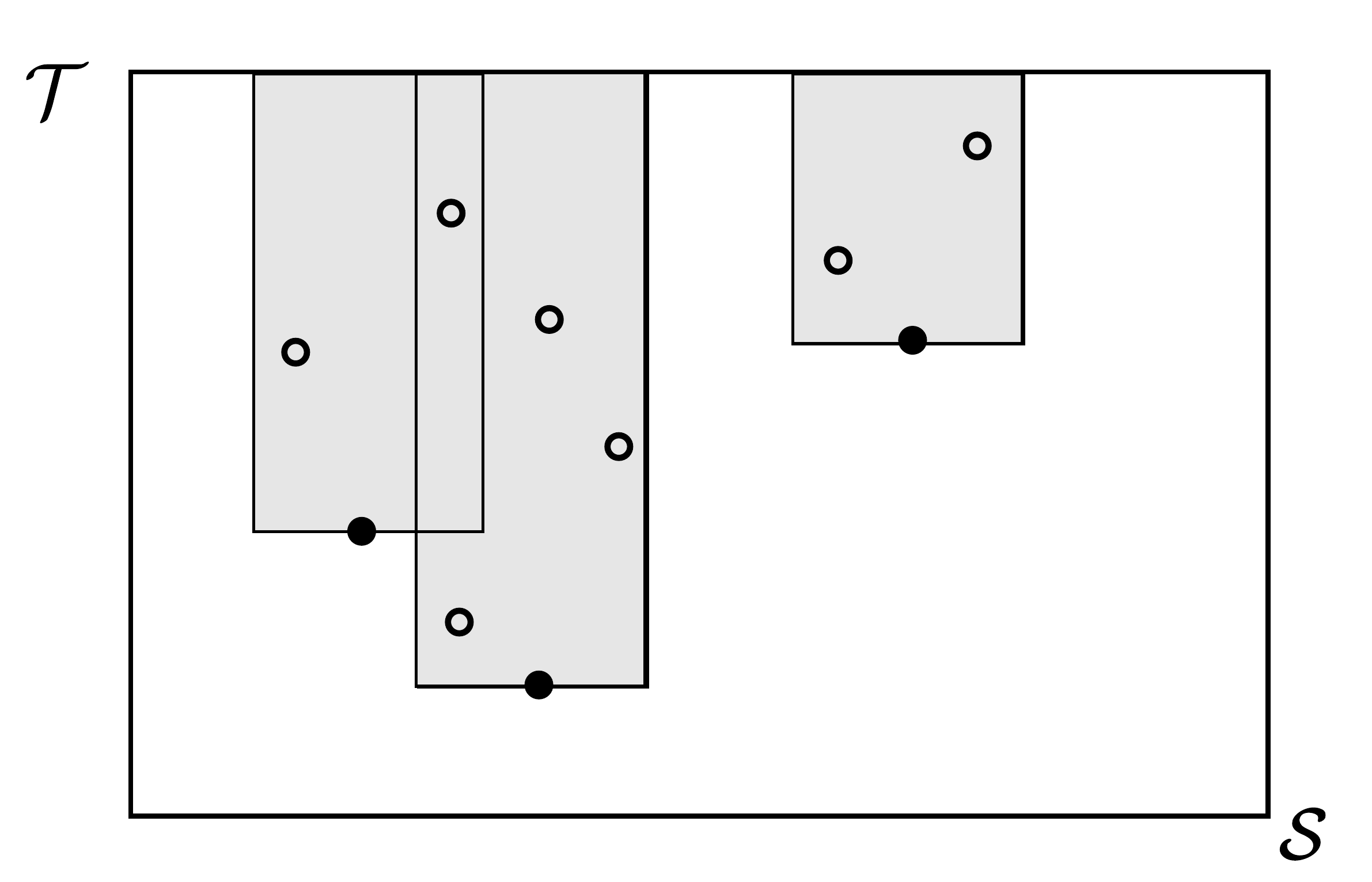}
    \includegraphics[width=.3\textwidth]{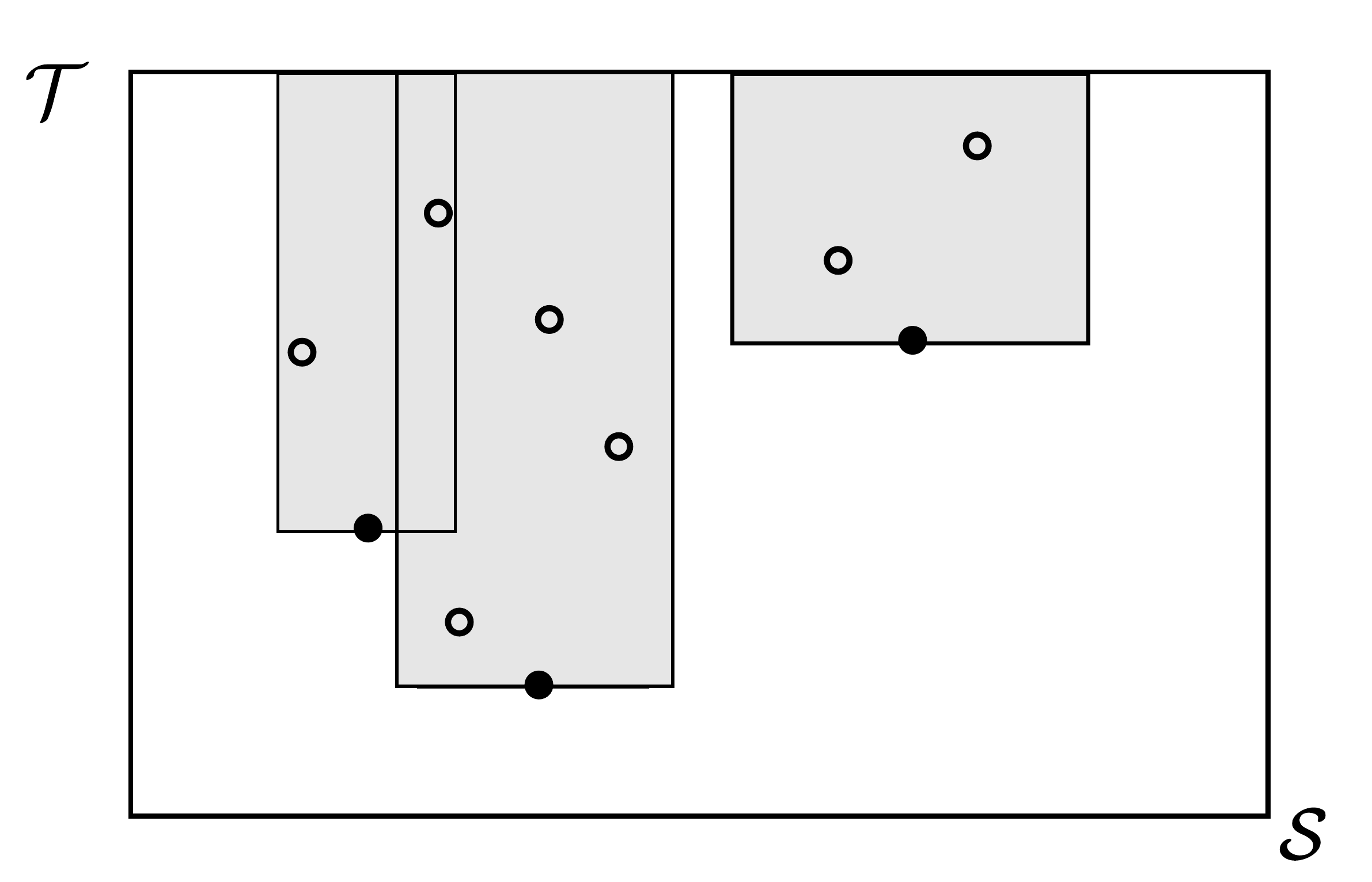}
    \includegraphics[width=.3\textwidth]{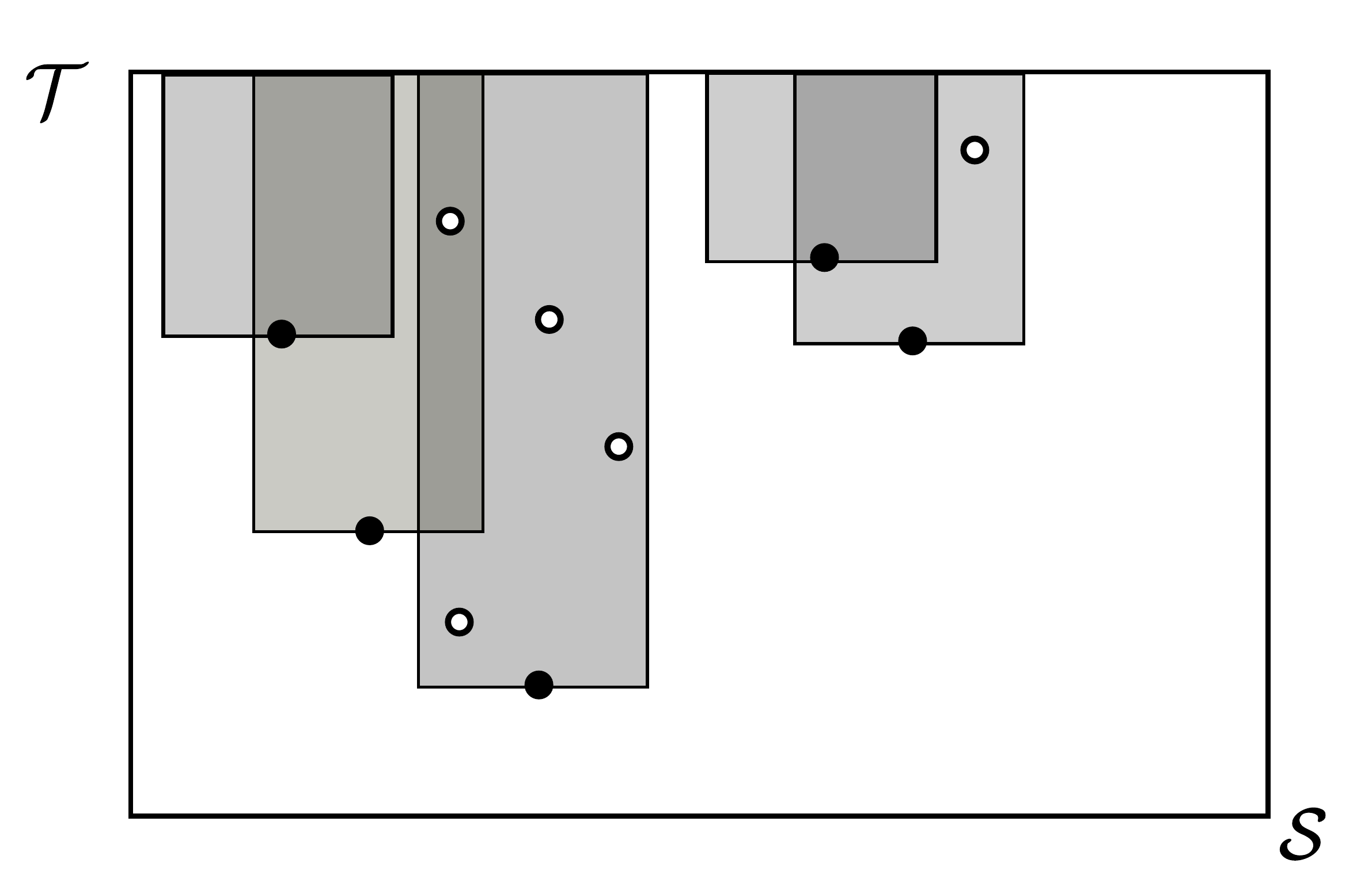}
 \caption[The \matern type-III hardcore point process]{(left) The \matern type-III hardcore point process in 1-d: filled dots (projected
  onto $\cS$) are \matern events, empty dots are thinned events.  The shaded
 region is the shadow.
(center) and (right): \matern type-III processes with varying radii and probabilistic deletion}
   \label{fig:matern3_fig}
   \end{figure}


Figure~\ref{fig:matern3_fig}(left) shows relevant events for the $1$-dimensional case. The filled dots
form the \matern process $G$ and the empty dots represent thinned events.
Both together form the primary process. 
Define the `shadow' of a point ${(s^*,t^*) \in \cS \times \cT}$ as the indicator function for all locations in ${\cS \times \cT}$ that would
be thinned by $(s^*,t^*)$. This is the set of all points whose $\cS$-coordinate is within $R$ of  $s^*$, and whose $\cT$-coordinate is greater than $t^*$. 
Letting $I$ be the indicator function, define the shadow 
of $(s^*,t^*)$ at $(s,t)$ as
\begin{align}
  \mathscr{H}(s,t;s^*, t^*, R) = I(t > t^*) I(\lVert s - s^* \rVert < R). 
\end{align}
The shadow of $G^+$ is all locations that would be thinned by any element of $G^+$: 
\begin{align}
  \mathscr{H}(s,t;G^+) = 1 - \!\!\!\!\!\!\! \prod_{(s^*,t^*) \in G^+} \!\!\!\!\!\!\! \left(1 - \mathscr{H}(s,t;s^*,t^*, R) \right).
\end{align}
For notational convenience, we drop the dependence of the shadow on $R$ above.
The shaded area in figure~\ref{fig:matern3_fig}(left) shows the shadow of all \matern events, $G^+$. Note that all 
thinned events must lie in the shadow of the \matern events, otherwise 
they couldn't have been thinned. Similarly, \matern events cannot lie in each
others shadows; however, they \emph{can} fall within the shadow of some thinned event.

The hardcore repulsive process can be generalized in a number of ways, 
For instance, instead of requiring all 
\matern events to have the same radius, we can assign each an independent radius drawn from 
some distribution $q(R)$. Such \matern processes are called softcore repulsive processes \citep{Hube:Wolp:2009}. In this case, the primary process
can be viewed as a Poisson process on a space whose coordinates are location $\mathcal{S}$, birth time $\mathcal{T}$ and radius
$\mathcal{R}$. Given a realization $F^+ \equiv (F, T^F, R^F)$, we define a secondary point process $G^+ \equiv (G, T^G, R^G)$ by deleting 
all points that fall within the radius associated with an older, undeleted primary event.  The locations $G$ constitute a sample from the 
softcore \matern type-III process. 
Given the triplet $(G, T^G, R^G)$, we can once again calculate the shadow $\mathscr{H}$, now defined as: 
\begin{align}
  \mathscr{H}(s,t;G^+) = 1 - \!\!\!\!\!\!\!\!\!\!\! \prod_{(s^*,t^*, r^*) \in G^+} \!\!\!\!\!\!\!\!\!\! \left( 1 - \mathscr{H}(s,t;s^*,t^*, r^*) \right).
\end{align}
Figure~\ref{fig:matern3_fig}(center) illustrates this; note that the radii of the thinned events are irrelevant. 

An approach to soft repulsion that we propose here is to \emph{probabilistically} thin events of the primary Poisson process. 
This is a flexible generalization of ideas present in the literature, allowing control over the strength of the repulsive effect as well as 
its span. The probability of deletion can be constant, or can depend on the distance of a point to a previous unthinned point, and a
primary event is retained only if it is left unthinned by all surviving points with earlier birth times. Write the deletion kernel associated
with location $s^*$ as $K(\cdot, s^*)$, so that the probability of thinning an event located at $(s,t)$ is given by $I(t > t^*)K(s, s^*)$. To keep this process efficient, 
one can use a deletion kernel with a compact support; figure~\ref{fig:matern3_fig}(right) illustrates the resulting shadow.
Where previously the \matern events defined a black-or-white shadow, now the shadow can have intermediate `grey' values corresponding
to the probability of deletion. Now the shadow at any location $(s,t) \in \cS \times \cT$ is given by
\begin{align}
  \mathscr{H}(s,t;G^+) &= 1 - \!\!\!\!\!\!\! \prod_{(s^*,t^*) \in G^+} \!\!\!\!\!\! \left(1 - I(t > t^*) K(s, s^*) \right). \label{eq:thin_shad}
\end{align}
Note that for this process, while the thinned events must still lie in the shadow $\mathscr{H}(s,t;G^+)$, \matern events \emph{can} lie in each other's shadow.
We recover the \matern processes with deterministic thinning by letting $K$ be the indicator function.

Another generalization is to allow the thinning probability to depend on the difference of the birth times of two events. This is useful in
applications where the repulsive influence of an event decays as times passes, and the thinning probability is given by~$K_1(t, t^*)K_2(s, s^*)$.
While we do not study this, we mention it to demonstrate the flexibility of the \matern framework towards developing realistic 
repulsive mechanisms.

Finally, we mention that repulsive processes on the real line can be viewed as generalized \matern type-III processes where birth-times are
observed. Probabilistic thinning recovers a class of self-inhibiting point processes commonly used to model neuronal spiking \citep{Brown2004a}. Similarly, renewal processes
can be viewed as a \matern type-III process where the shadow $\mathscr{H}$ has a special Markovian construction. 

\subsection{Probability density of the \matern type-III point process}  \label{sec:matern_pdf}

%
Let $(\cS, \Sigma)$ be a subset of a $d$-dimensional Euclidean space ($\Sigma$ is the Borel $\sigma$-algebra). For two points $\bs_1, \bs_2 \in \cS$,
define ${\bs_1 > \bs_2}$ if ${s_{1d} > s_{2d}}$. Thus we use the last coordinate to define a partial ordering on $\cS$.
A realization 
of a \matern process on $\cS$  is obtained by thinning a primary Poisson process, and has finite cardinality if the primary process is finite. 
We restrict ourselves to this case. Note that with probability one, there is unique ordering of the Poisson elements
according to the partial ordering we defined earlier. This will allow us to associate realizations of point processes
with unique, increasing sequences of points in $\cS$. Below, we describe this space of point process realizations.

For each~$n$, let~$\cS^n$ be the~$n$-fold product space with the usual product~$\sigma$-algebra,~$\Sigma^n$. We refer
to elements of~$\cS^n$ as~$S^n$. 
Define the union space ${\tilde{S}^{\cup} \equiv \bigcup_{n=0}^{\infty} \cS^n}$ (where 
$\cS^0$ is a point satisfying ${\cS^0 \times \cS = \cS \times \cS^0 = \cS}$) and assign it 
the~$\sigma$-algebra $\tilde{\Sigma}^{\cup} \equiv \{\bigcup_{n=0}^{\infty} A^n, \ \ \forall A^n \in \Sigma^n\}$. 
$\tilde{S}^{\cup}$ is the space of all finite sequences in~$\cS$, and define $(\Su, \Sigma^{\cup})$ as its restriction to {increasing} sequences
in $\cS$. We treat finite point processes as random variables taking values in $\Su$, and refer to elements of this space
by uppercase letters (e.g.\ $S$). 
For a finite measure $\mu$ on $(\cS,\Sigma)$ (e.g.\ Lebesgue measure), let $\mu^n$ be the $n$-fold product measure on~$(\cS^n, \Sigma^n)$.
Assign any set $B \in \Sigma^{\cup}$ the measure
\begin{align}
  \muu(B) &= \sum_{n=0}^{\infty} \mu^n(B \cap \cS^n)
         = \sum_{n=0}^{\infty} \int_{B \cap \cS^n}  \mu^n(\dif S^n). \label{eq:base_measure1}
\end{align}
Now, let $S$ be a realization of a Poisson process with mean measure $\Lambda$, and assume $\Lambda$ 
admits a density $\lambda$ with respect to $\mu$. 
Writing $|S|$ for the cardinality of $S$, we have:
\begin{thrm} \emph{(Density of a Poisson process)} \label{thrm:poiss_density}
  A Poisson process on the space $\mathcal{S}$  with intensity $\lambda(s)$ (where $\Lambda(\cS) = \int_{\cS} \lambda(s) \mu(\mathrm{d} s) < \infty$) is a random variable taking values in $(\Su, \Sigma^{\cup})$ with 
probability 
density w.r.t.\ the measure $\muu$
given by 
\vspace{-.05in}
\begin{align}
  p(S) = \exp(-\Lambda(\cS)) \prod_{j=1}^{|S|} \lambda(s_j)  \label{eq:poiss_prob}
\end{align}
\end{thrm}
\vspace{-.02in}
For a proof, see \cite{DalVer2008a}. The density $p(S)$ is called the \emph{Janossy density}, and is 
used in the literature to define the Janossy measure. This measure is not a probability measure, however 
our definition of $S$ as an ordered sequence ensures 
\begin{align}
  \int_{\Su} p(S) \muu(\dif S) &= 1.
\end{align}
We now return to the \matern type-III process. Recall that events of the augmented primary Poisson process ${F^{+} = (F, T^F)}$ lie in the product space
$(\cS \times \cT)$, where $\cT$ is the unit interval. 
Since we order elements by their last coordinate, $F^+$ is a sequence with increasing birth-times $T^F$.
Let $\mu$ be a measure on $(\cS \times \cT)$; when $\cS$ is a subset of $\mathbb{R}^2$, $\mu$ is just 
Lebesgue measure on $\mathbb{R}^3$. 
By first writing down the density of the augmented primary Poisson process $F^+$, and then using the thinning construction of the 
\matern type-III process, we can calculate the probability density of the augmented \matern type-III process ${G^{+} = (G, T^G)}$ with respect to the 
measure ${\mu}^{\cup}$:
\begin{thrm}  \label{thrm:mat_dens} Let $G^+ = (G, T^G)$ be a sample from a generalized \matern type-III process, augmented with the birth 
  times. Let $\lambda$ be its intensity, and $\mathscr{H}(s,t;G^+)$ be its shadow following the appropriate thinning scheme. 
Then, its density w.r.t.\  ${\mu}^{\cup}$ is 
\begin{align}
 p(G^+ \given \lambda) &=\exp\left(-\lambda \int_{\cS \times \cT}\left( 1 - \mathscr{H}(s,t;G^+)\right)\mu(\dif s \,\dif t)\right)
              \lambda^{|G^+|}  \nonumber \\
         & \quad \times  \prod_{g^+ \in G^+} \left( 1 - \mathscr{H}(g^+;G^+) \right).
\label{eq:mat_marg_prob}
\end{align}
\end{thrm}
We include a proof in the appendix. 
The product term in the expression above penalizes \matern events that fall within the shadow of earlier events; in fact, for deterministic thinning, such
an occurrence will have zero probability. The exponentiated integral encourages the shadow to be large, which in turn implies that the events are 
spread out. 

A similar result was derived in \cite{Hube:Wolp:2009}; they express the \matern type-III density with respect to a homogeneous Poisson 
process with unit intensity. However, their result applied only to the hardcore process. Also, their proof is less direct than ours, proceeding 
via a coupling from the past construction. It is not clear
how such an approach extends to the more complicated extensions we introduced above.

We now have the following corollary of the previous theorem:
\begin{coro} \label{prop:mat_post} Let ${G^+ = (G, T^G)}$ be a sample from a \matern type-III process 
 augmented with its birth times. Let $\lambda$ be its intensity, and $\mathscr{H}(s,t;G^+)$ its shadow.
Then, given $G^+$, the conditional distribution of the locations and birth times of the thinned events
${\tG^+ = (\tG,T^{\tG})}$ is a Poisson process on ${\cS \times \cT}$, with intensity $\lambda  \mathscr{H}(s,t; G^+)$.
\end{coro}
\begin{proof}
Observe that the primary Poisson process $F^+ = (\tG^+ \cup G^+)$ (where the union of two ordered sequences is their concatenation followed
by reordering).  The joint $p(\tG^+, G^+)$ is the density of $F^+$ multiplied by the
probability that the elements of $F^+$ are assigned labels `thinned' or `not thinned'. 
The latter depends on the shadow $\mathscr{H}$, and it follows easily from Theorem~\ref{thrm:poiss_density} 
(see equation~\eqref{eq:mat_joint_app} in the appendix) that:
\begin{align}
   p(\tG^+, G^+) &= \exp(-\lambda \mu(\cS \times \cT)) \lambda^{n}  \! \! \! \!\! \prod_{(s,t) \in \tG^+}  \! \! \! \!\mathscr{H}(s,t;G^+)  \! \! \! \! \! \prod_{(s,t) \in G^+}  \! \! \! \!(1- \mathscr{H}(s,t;G^+)). \label{eq:mat_joint}
\end{align}
Plugging this, and equation \eqref{eq:mat_marg_prob} into Bayes' rule, we have
\begin{align}
  p(\tG^+ | G^+) &= \frac{p(\tG^+, G^+)}{p(G^+)} \\ 
  &= \exp \left(- \int_{\cS \times \cT} \!\!\!\lambda\, \mathscr{H}(s,t; G^+) \mu(\dif s\ \dif t) \right)
              \prod_{(\tilde{s},\tilde{t}) \in \tG^+} \lambda \mathscr{H}(\tilde{s}, \tilde{t}; G^+).  \label{eq:mat_post}
\end{align}
From Theorem~\ref{thrm:poiss_density}, this is the density of a Poisson process with intensity $\lambda \mathscr{H}(s, t; G^+)$. \hfill ${}_\blacksquare$ 
\end{proof}

The result above provides a remarkably simple characterization of the thinned primary events. 
Rather than having to resort to incremental birth-death schemes that update the thinned Poisson events one at a time, we can 
jointly simulate all of these from a Poisson process, directly obtaining their number and locations. 
Such an approach is much simpler and much more efficient, and it is
central to the MCMC sampling algorithm  we describe in the next section.

The intuition behind this result is that for a type-III process, a point of the primary Poisson process $F$ can
be thinned only by an element of the secondary process. Consequently, given the
secondary process, there are no interactions between the thinned events themselves: given the \matern process, the thinned events are just 
Poisson distributed. Such a strategy does not extend to
\matern type-I and II processes where the fact that thinned events \emph{can} delete each other means that the posterior is
no longer Poisson. For instance, for any of these processes, it is not possible for a thinned event to occur by itself within any neighbourhood
of radius $R$ (else it couldn't have been thinned in the first place). However, two or more events \emph{can} occur together. Clearly such a process is not 
Poisson, rather it possesses a clustered structure.

\section{Bayesian modeling and inference for \matern type-III processes} \label{sec:inf_mat}


  In the following, we model an observed sequence of points $G$ as a realization of a \matern type-III process. The parameters governing this are
the intensity of the primary process, $\lambda$, and the parameters of the thinning kernel, $\theta$. For the hardcore process, $\theta$ is just
the interaction radius $R$, while with probabilistic thinning, $\theta$ might include an interaction radius $R$, and a thinning
probability $p$ (with $p = 1$ recovering the hardcore model). 
For the softcore process, each \matern event has its own interaction radius which we have to
infer, and $\theta$ would be this collection of radii. In this case we might also assume that the distribution these radii are drawn from 
has unknown parameters.

Taking a Bayesian approach, we place priors on the unknown parameters. A natural prior for $\lambda$ is the conjugate Gamma density. The Gamma is also a 
convenient and flexible prior for the thinning length-scale parameter $R$. For the case of probabilistic thinning where $\theta = (p,R)$, we can place a Beta prior on the 
thinning probability $p$. For the softcore model, we model the radii as uniformly distributed over $[r_L, r_U]$, and place conjugate
hyperpriors on $r_L$ and $r_U$. For simplicity, we leave out any hyperparameters in what follows, and writing $q$
for the prior on $\theta$, we have
\begin{align}
  \lambda &\sim \text{Gamma}(a,b) \\
  \theta  &\sim q \\
  F^+ \equiv (F,T^F) & \sim \text{Poisson Process}(\lambda) \\
  G^+ \equiv (G,T^G) & \sim \text{Thin}(F^+, \theta)
\end{align}
Note that $G^+$ includes the \matern events $G$ as well as their birth times $T^G$; however, we only observe $G$. Given $G$, we require the posterior 
distribution $p(\lambda, \theta \given G)$. 
We will actually work with the augmented posterior $p(\lambda, \theta, F^+, T^G \given G)$. In particular, we set up a Markov chain whose state 
consists of all these variables, and whose transition operator is a 
sequence of Gibbs steps that conditionally update 
each of these four groups of variables. We describe the four Gibbs updates below.

\subsection{Sampling the thinned events}
Given the \matern events $G^+$ and the thinning parameters $\theta$, we can calculate the shadow $\mathscr{H}_{\theta}(s,t;G^+)$ 
(here we make explicit the dependence of $\mathscr{H}$ on~$\theta$).
Sampling the thinned events $\tG^+$ is now a straightforward application of Corollary~\ref{prop:mat_post}:
discard the old events, and simulate a new sequence from a Poisson process with intensity $\lambda \mathscr{H}_{\theta}(s,t;G^+)$. 

We do this by applying the thinning theorem for Poisson processes (\cite{Lewis1979}, see also Theorem~\ref{thrm:Thin}). 
In particular, we first sample 
a homogeneous Poisson process with intensity $\lambda$ on $(\cS \times \cT)$ and keep each point $(s,t)$ of this process with probability 
$\mathscr{H}_{\theta}(s, t;G^+)$. The surviving points form a sample
from the required Poisson process. For models with deterministic thinning, $\mathscr{H}_{\theta}$ is a binary function, and the posterior is just a Poisson
process with intensity $\lambda$ restricted to the shadow. 
Note that this step eliminates the need for any birth-death steps, and provides a simple and global way to vary the number of thinned events from 
iteration to iteration. 

\subsection{Sampling the birth times of the \matern events}
From equation \eqref{eq:mat_joint}, we see that the birth times $T_G$ of the \matern events have density
\begin{align}
 p(T^G | G, F^+, \lambda, \theta) &\propto 
              \prod_{({s}, {t}) \in G^+} \left(1 - \mathscr{H}({s},{t};G^+) \right)
              \prod_{({s}, {t}) \in \tG^+} \mathscr{H}({s},{t};G^+)  \label{eq:birth_joint}
\end{align}
A simple Markov operator that maintains equation \eqref{eq:birth_joint} as its stationary distribution is a Gibbs sampler that
updates the birth times one at a time. For each \matern event~${g \in G}$, we look at all primary events (thinned or not) within
distance $r^g$ (where $r^g$ is the interaction radius associated with $g$). The birth times of these events segment the unit interval into a number of regions, 
and the birth time $t^g$ of $g$ is uniformly distributed within each interval (since as $t^g$ moves over an interval, the shadow at all 
primary events remains unchanged). 
As $t^g$ moves from one segment to the next, one of the primary events moves into or out of the shadow of $g$. The probability of any interval is then proportional to
the probability that these neighboring events are assigned their labels `thinned' or `not thinned' under the shadow that results when $g$ is assigned to that 
interval; this is easily calculated for each thinning mechanism.
For instance, in the \matern process with deterministic thinning, 
the birth time $t^g$ is uniformly distributed on the interval $[0, t_{min}]$, 
where $t_{min}$ is the time of the oldest event that is thinned only by $g$ ($t_{min}$ equals $1$ if there is no
such event). 

While it is not hard to develop more global moves, we found it sufficient to sweep through the \matern events, 
sequentially updating their birth times. This, together with jointly updating all thinned event locations and birth times was 
enough to ensure that the chain mixed rapidly.

\subsection{Sampling the Poisson intensity} \label{sec:Poiss_int_inf}
Having reconstructed the thinned events, it is easy to resample the intensity  $\lambda$.
Note that the number of primary events $|F|$ is Poisson distributed with intensity $\lambda$. 
With a conjugate Gamma$(a,b)$ prior on the Poisson intensity $\lambda$, 
the posterior is also Gamma distributed, with parameters 
$a_{post} = a + |F|$,
$\frac{1}{b_{post}} = \frac{1}{b} + \frac{1}{\mu(\cS)}$. 

\subsection{Sampling the thinning kernel parameter $\theta$}
Like the birth times $T^G$, the posterior distribution of $\theta$ follows from equation \eqref{eq:mat_joint}. For a prior $q(\theta)$, 
the posterior is just 
\begin{align}
 p(\theta | G^+, F^+, \lambda) &\propto  q(\theta) 
           \!\!\!   \prod_{({s}, {t}) \in G^+}\!\!\! \left( 1 - \mathscr{H}_{\theta}({s},{t};G^+)\right)  
           \!\!\!   \prod_{({s}, {t}) \in \tG^+}\!\!\! \mathscr{H}_{\theta}({s},{t};G^+)  
\end{align}
Again, sampling $\theta$ is equivalent to sampling a latent variable in a two-class classification problem,
with the \matern events and thinned events corresponding to the two classes. Different values of $\theta$ result in different shadows
$\mathscr{H}_{\theta}(\cdot; G^+)$, and the likelihood of $\theta$ is determined by the probability of the labels
under the associated shadow.
For the models we consider, this results in a simple piecewise-parametric posterior distribution.

For the \matern hardcore process, ${\theta = R}$ is the interaction radius, whose posterior distribution is a truncated version of the prior.
The lower bound requires that no thinned event lies outside the new shadow, while the upper bound
requires that no \matern event lies inside the shadow. 
Sampling from this is straightforward. 
The same applies for the softcore model, only now, each \matern event has
its own interaction radius, and we can conditionally update them. {Given the set of radii, updating any hyperparameters of $q$ is easy}.
Finally, for the model with probabilistic thinning, we have $\theta = (p,R)$. To simulate $R$, we segment the positive real line into a finite number of segments, with
the endpoints corresponding to values of $R$ when a primary event moves into the shadow of a \matern event. Over each segment, the likelihood remains constant.
It is a straightforward matter to sample a segment, and then conditionally simulate a value of $R$ within that segment.
To simulate $p$, we simply count how many opportunities to thin events were taken or missed, and with a Beta prior on $p$, the posterior follows easily.

\section{Experiments}
  \begin{figure}
  \begin{minipage}[h]{0.32\linewidth}
  \centering
    \includegraphics[width=0.99\textwidth]{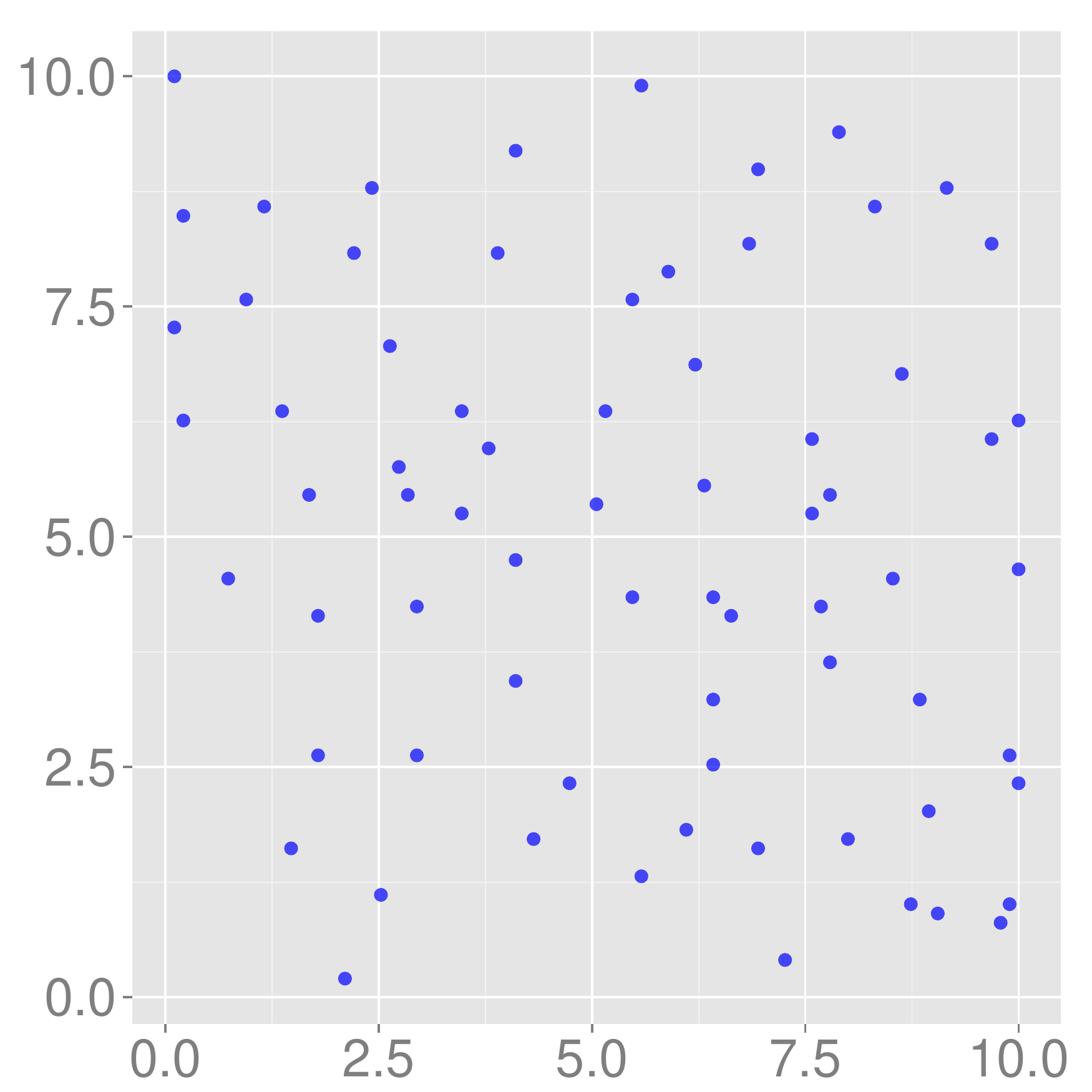} 
  \end{minipage}
  \begin{minipage}[h]{0.32\linewidth}
  \centering
  \includegraphics[width=0.99\textwidth]{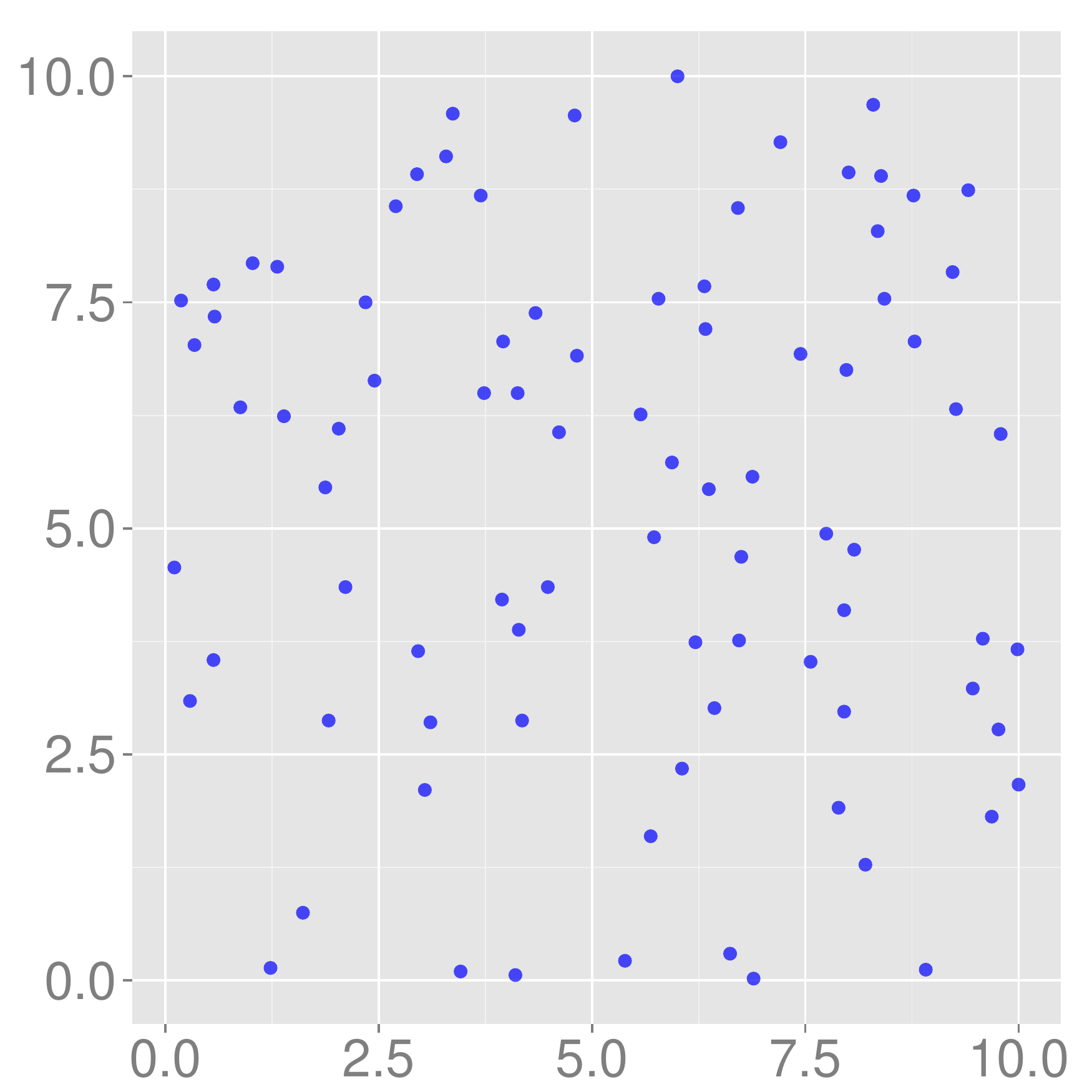}
  \end{minipage}
  \begin{minipage}[h]{0.32\linewidth}
  \centering
  \includegraphics[width=0.99\textwidth]{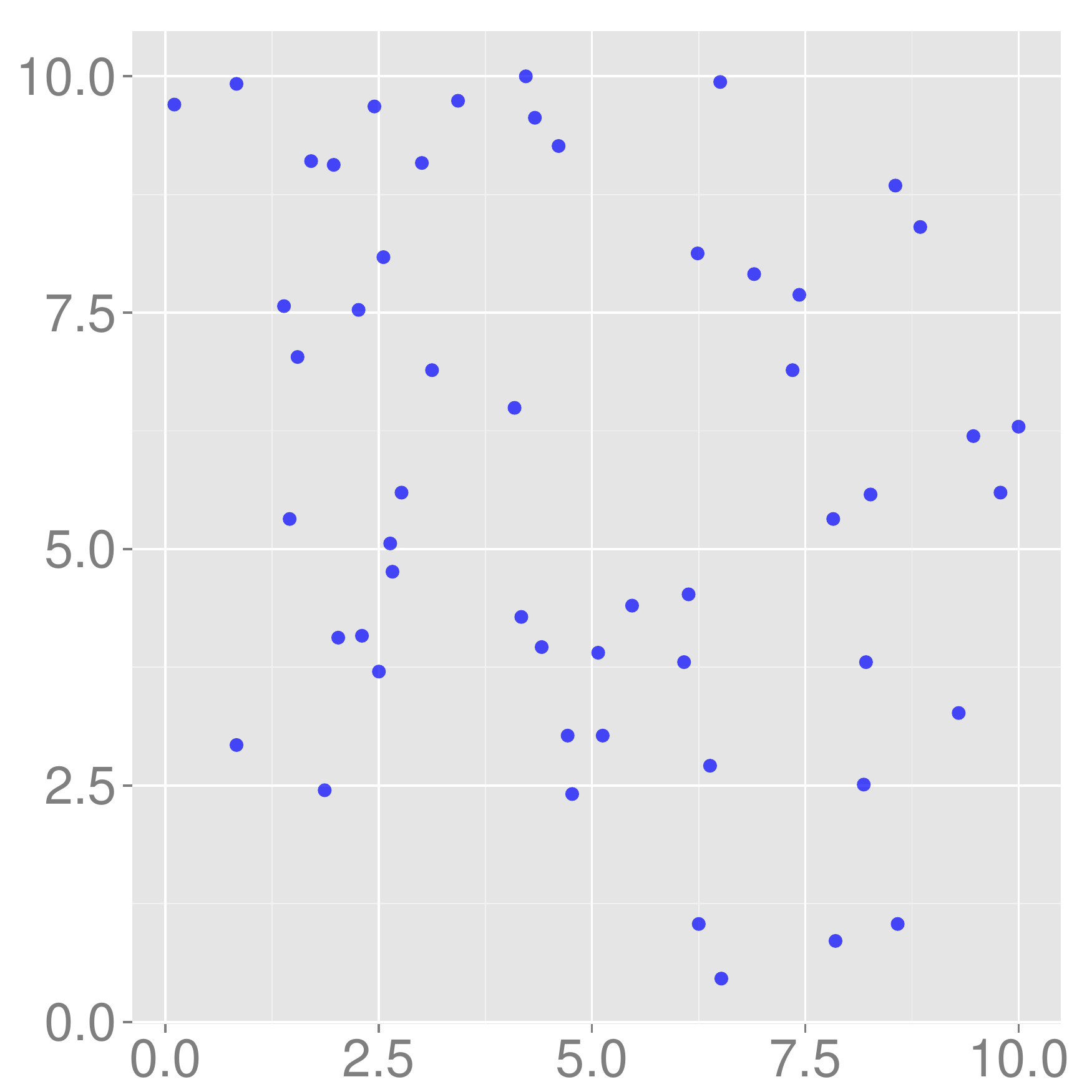}
  \end{minipage}
  \caption[Post Pred]{(Left): Swedish pine tree dataset. (Center and right): Nerve fiber entry points for mild and moderate neuropathy.}
  \label{fig:diab}
  \end{figure}
  We start with the classic Swedish pine tree dataset \citep{Ripley88},
available as part of \texttt{R} package \texttt{spatstat} \citep{spatstat}. 
This consists of the locations of $71$ trees which we normalized to lie in a $10$-by-$10$ square (see figure~\ref{fig:diab}(left)), 
and which we modeled as a realization of a \matern type-III hardcore point process. 
We placed a $\text{Gamma}(1,1)$ prior on the intensity $\lambda$ (noting there are about $100$ points in a $10\times10$ square), and a flat prior on 
the radius $R$ (noting that $R$ cannot exceed the minimum separation of the \matern events). 
All results were obtained from $10000$ iterations of our MCMC sampler after discarding the first
$1000$ samples as burn-in. 
A Matlab implementation of our sampler 
on an Intel Core 2 Duo $3$Ghz CPU took around a minute to produce $10000$ MCMC samples.
To correct for correlations across MCMC iterations, and to assess mixing for our MCMC sampler, we used software from \cite{Rcoda2006} to
estimate the effective sample sizes (ESS) of the various quantities; this gives the number of independent samples with the same information content as
the MCMC output.
Table~\ref{tab:ess_mat_hc} shows these values, suggesting our sampler mixes rapidly.

  \begin{figure}
  \begin{minipage}[h]{0.233\linewidth}
  \caption[Posterior distributions of the homogeneous \matern type-III hardcore model for the Swedish pine tree dataset]
  {Swedish tree dataset:  \matern hardcore posteriors of intensity, thinning radius, and number of thinned events}
  \label{fig:swedish_hc_post}
  \end{minipage}
  \begin{minipage}[h]{0.25\linewidth}
  \centering
    \includegraphics[width=0.99\textwidth]{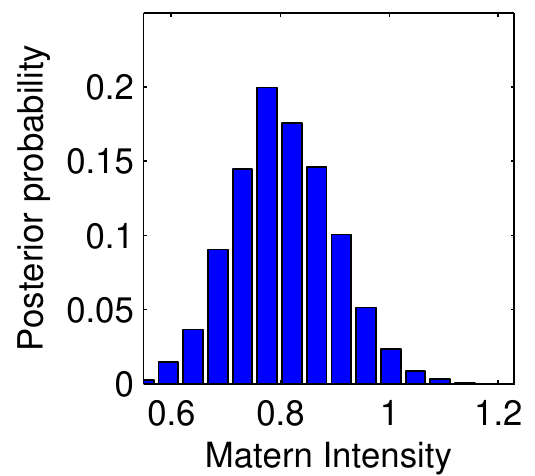} 
  \end{minipage}
  \begin{minipage}[h]{0.25\linewidth}
  \centering
    \includegraphics[width=0.99\textwidth]{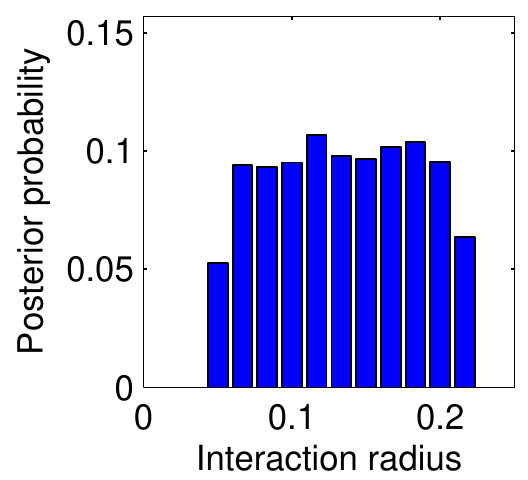} 
  \end{minipage}
  \begin{minipage}[h]{0.25\linewidth}
  \centering
    \includegraphics[width=0.99\textwidth]{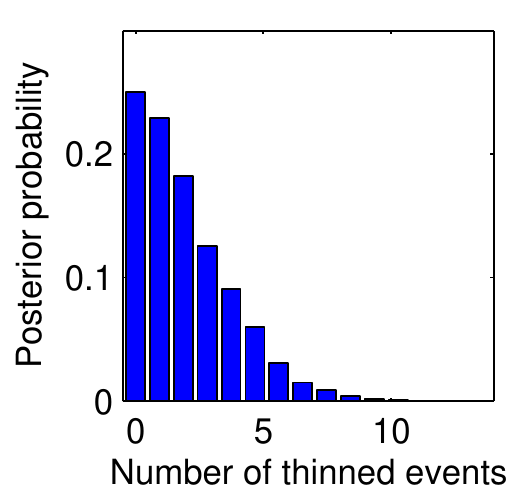} 
  \end{minipage}
  \end{figure}

The left and center plots in figure~\ref{fig:swedish_hc_post} show the posterior distributions over $\lambda$ and $R$ 
for the Swedish dataset.
Recall that the area of the square is $100$, while the dataset has $71$ \matern events.
The fact that the posterior over the intensity $\lambda$ concentrates around 
$0.8$ suggests that the number of thinned events is small. This is confirmed by the rightmost plot 
of figure~\ref{fig:swedish_hc_post} which shows the posterior distribution over the number of thinned  points is less than $10$.
\begin{table}
  \caption{\label{tab:ess_mat_hc}Effective sample sizes (per 1000 samples) for the \matern type-III process } 
\centering
\fbox{%
\begin{tabular}{|c|c|c|}
\hline
 & Swedish pine dataset & Nerve fiber data\\
 & (Hardcore model) & (Probabilistic thinning)\\
\hline
\matern interaction radius &  $344.51$ & $121$\\
\hline
Latent times (averaged across observations) & $989.4$ &  $880.1$\\
\hline
Primary Poisson intensity & $954.7$ & $680.2$ \\
\hline
\end{tabular}}
\end{table}
While the data in figure~\ref{fig:diab} are clearly underdispersed, few thinned events suggests a weak repulsion.
This mismatch is a consequence of occasional nearby events in the data, and the fact that the hardcore model requires
all events to have the same radius. Consequently, the interaction radius must be bounded by the minimum inter-event distance.
 Figure~\ref{fig:swedish_hc_post} (middle) shows that under the posterior, the interaction radius is bounded around $0.2$, much less than
 the typical inter-event distance.

  Figure~\ref{fig:swed_L_pred} quantifies this lack of fit, 
  showing Besag's L-function $L(r)$ \citep{Besag77}.
  For convenience, by $L(r)$ we mean~${L(r) - r}$, which measures the excess number events (relative to a Poisson process) within 
  a distance $r$ of an arbitrary event. A Poisson process has ${L(r) = 0}$, while
  ${L(r) < 0}$ indicates greater regularity than Poisson at distance $r$. 
The continuous magenta line in the subplots shows the empirical $L(r)$ as a function of distance $r$ for the 
  pine tree dataset. 
  We see that the data are much more regular than Poisson, especially over distances between $0.5$ to $1.5$ where a clear repulsive trend is seen.
  The blue envelope in the top left subplot shows the posterior predictive distribution of the L-function 
  from the type-III hardcore process. For this, at each MCMC iteration we generated a new realization of the hardcore process with the 
  current parameter values, and estimated $L(r)$ for that. The dashed line is the mean $L(r)$ across MCMC samples, while the two bands show $25\%$-$75\%$ and
  $2.5\%$-$97.5\%$ values. 
  While the data are clearly repulsive, the predictive intervals for the hardcore process in figure~\ref{fig:swed_L_pred} do not match the trend observed in the data, 
  corresponding instead to a Poisson process with a small $R$. 

  \begin{figure}
  \begin{minipage}[h]{0.22\linewidth}
  \caption[Post Pred]{Posterior predictive values of L-functions for the Swedish pine tree dataset. (Top left) is the \matern hardcore model,
  (top right) is the softcore model, and (bottom left) is probabilistic thinning. (Bottom right): predictive values for a Strauss process fit.}
  \label{fig:swed_L_pred}
  \end{minipage}
  \begin{minipage}[h]{0.76\linewidth}
  \begin{minipage}[h]{0.49\linewidth}
  \centering
  \includegraphics[width=0.99\textwidth]{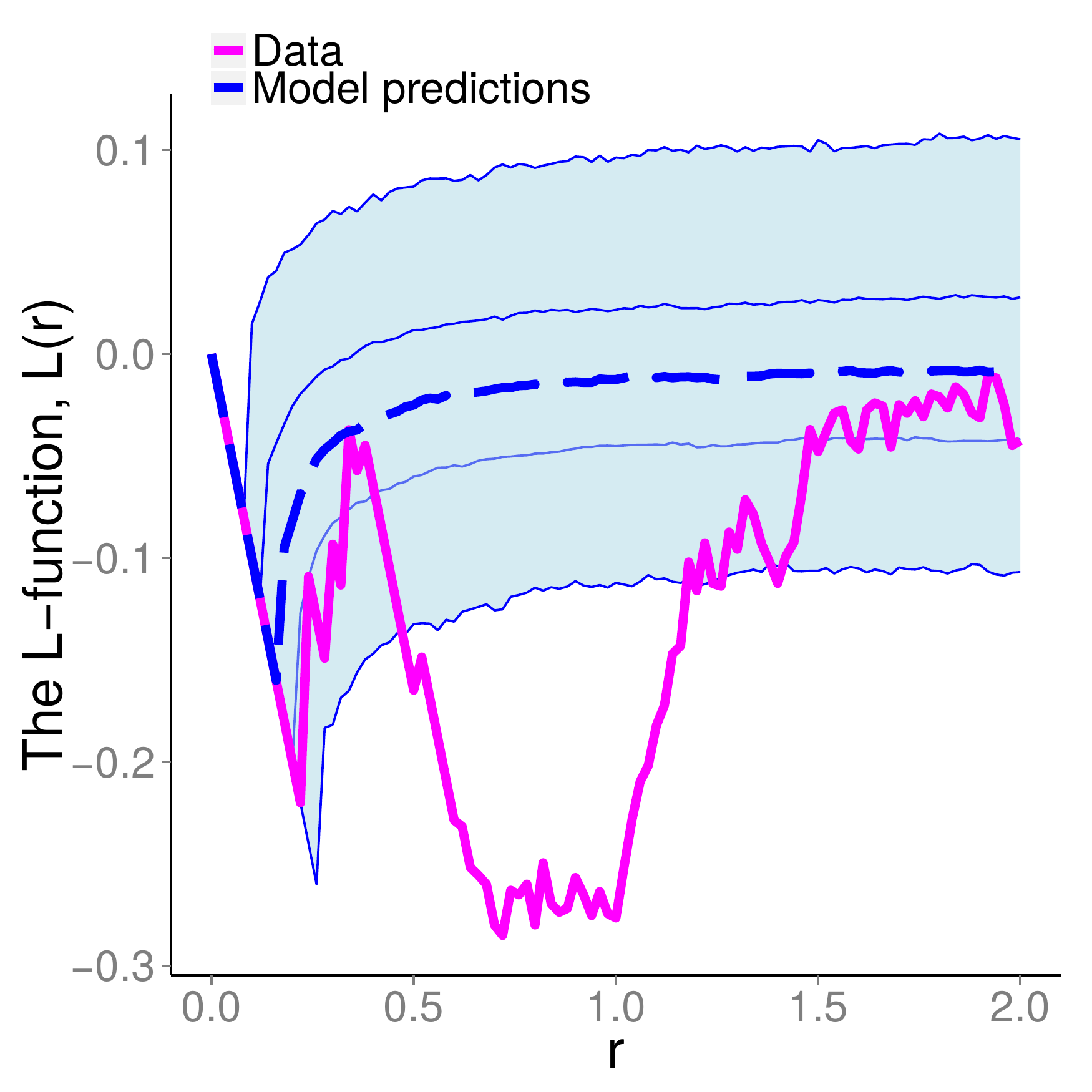}
  \end{minipage}
  \begin{minipage}[h]{0.49\linewidth}
  \centering
  \includegraphics[width=0.99\textwidth]{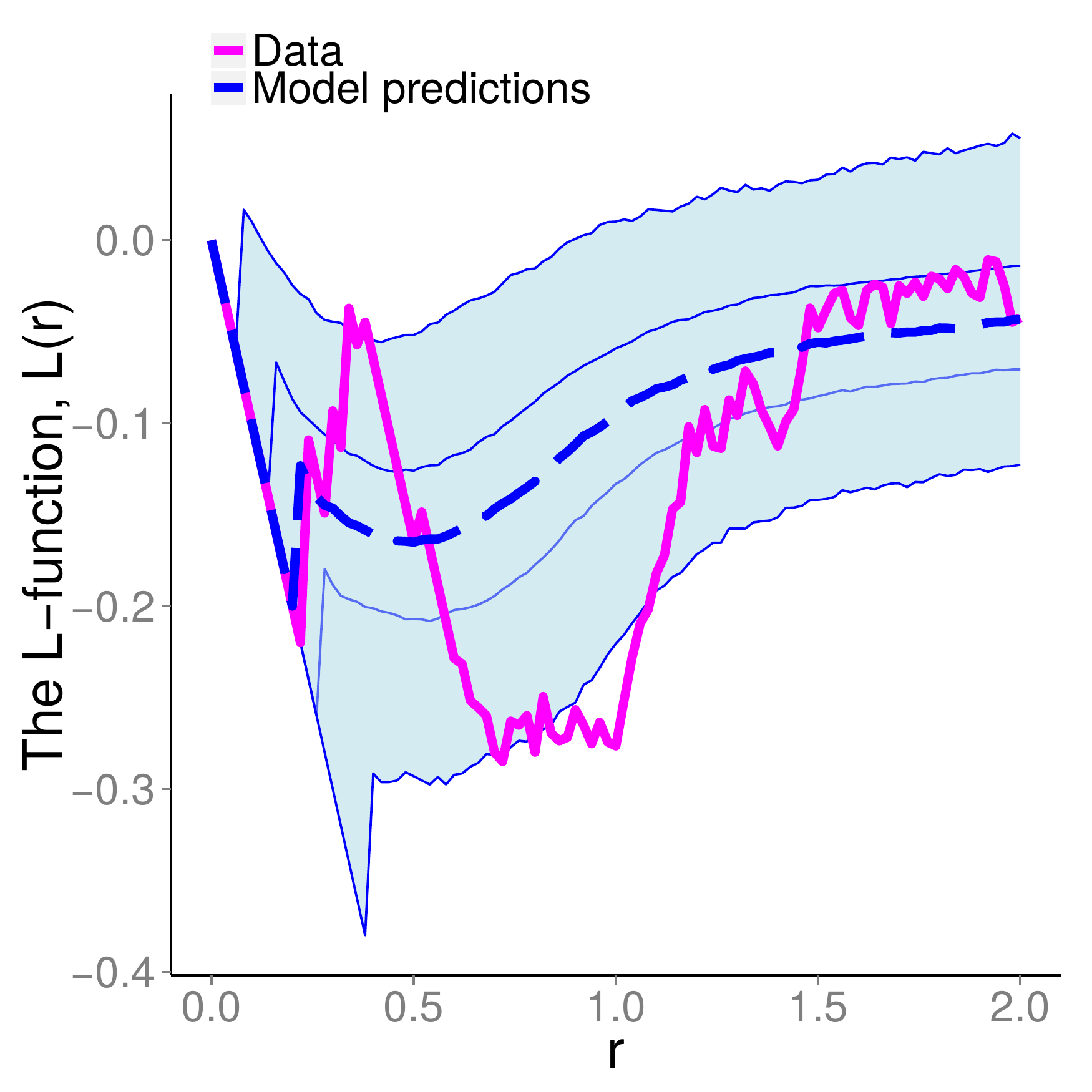}
  \end{minipage}
  \begin{minipage}[h]{0.49\linewidth}
  \centering
  \includegraphics[width=0.99\textwidth]{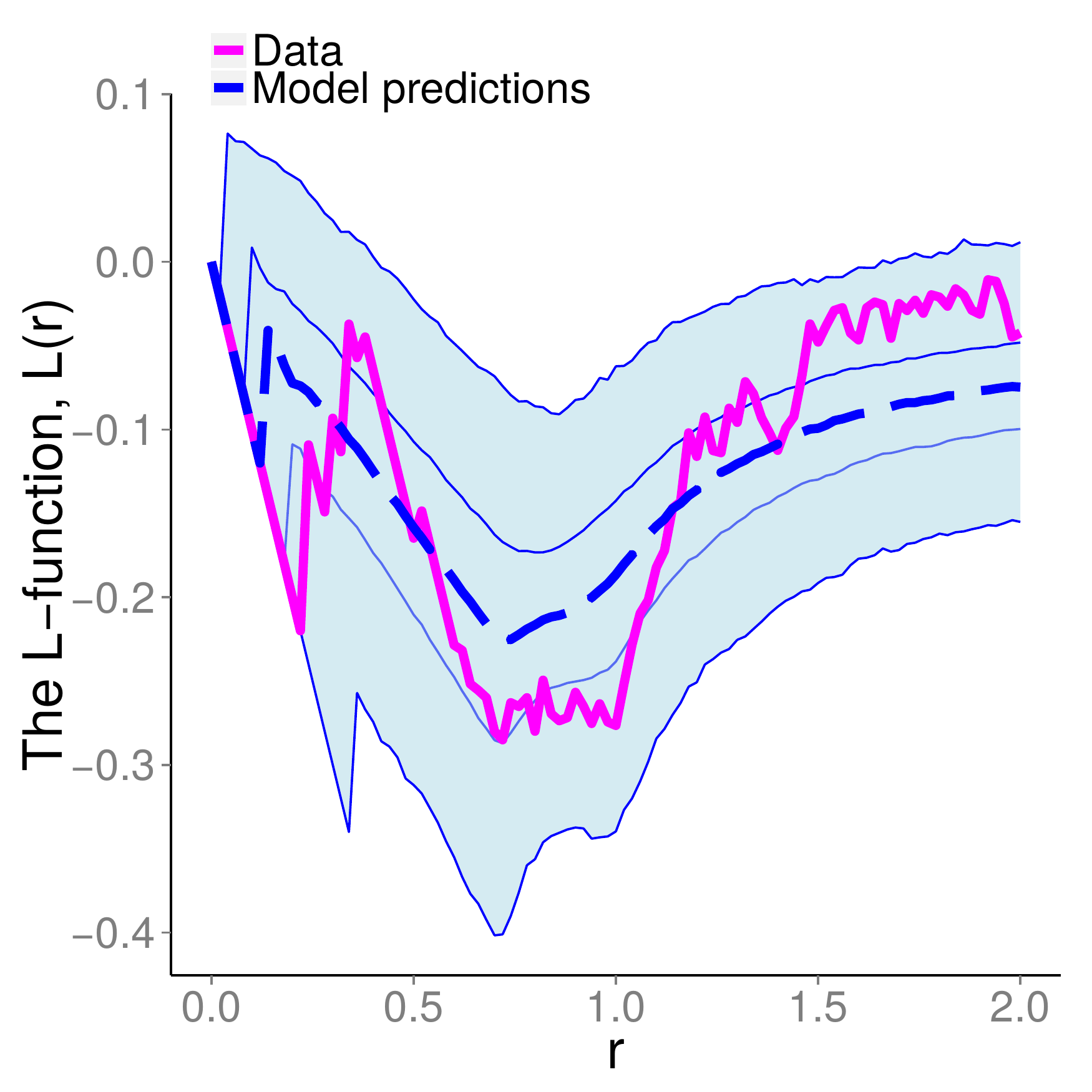}
  \end{minipage}
  \begin{minipage}[h]{0.49\linewidth}
  \centering
  \includegraphics[width=0.98\textwidth]{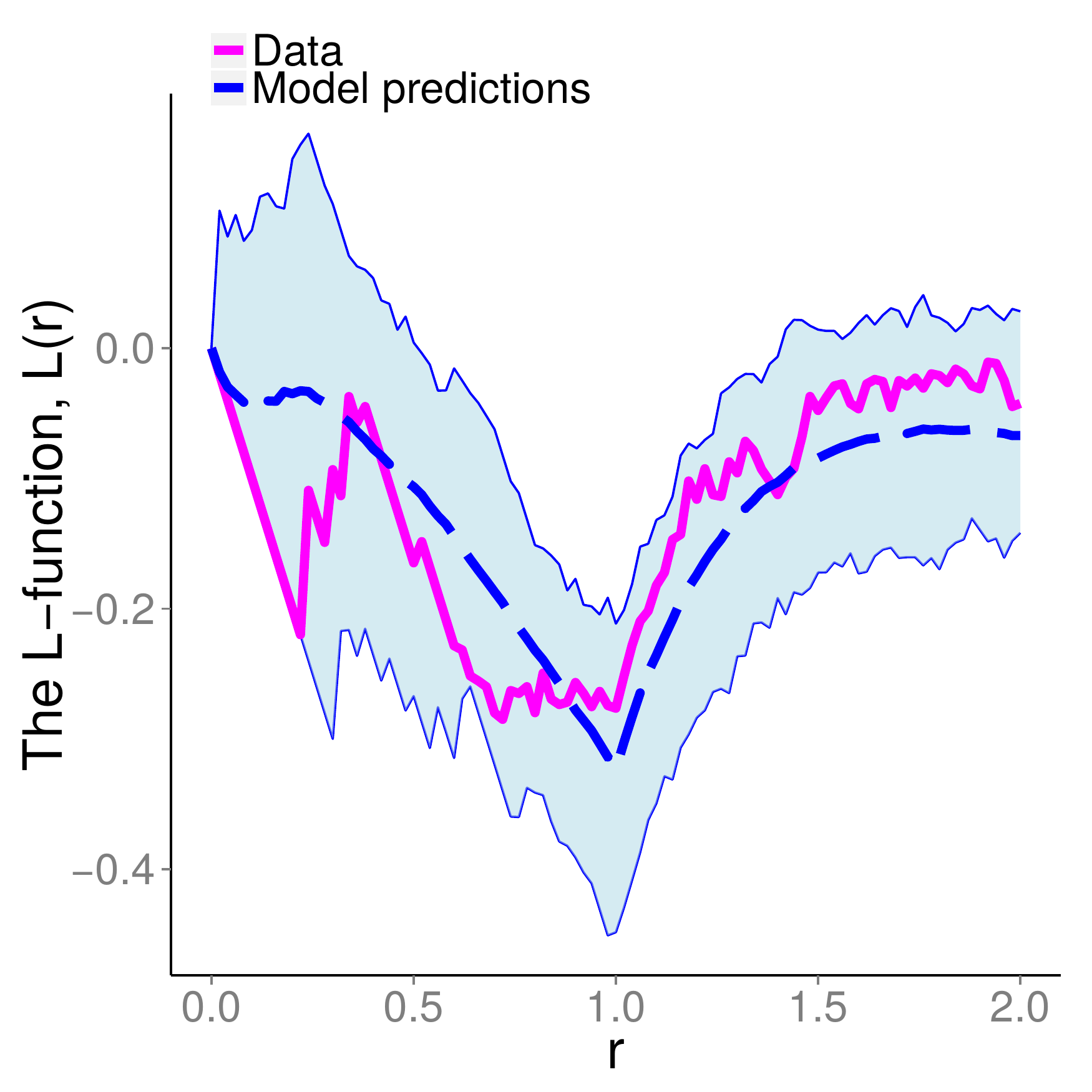}
  \end{minipage}
  \end{minipage}
  \end{figure}
We next fit the data using the softcore model 
(where each \matern event has its own interaction radius). 
We allowed the radii to take values in the interval $[r_L,r_U]$, and placed conjugate priors on these limits of this interval.
Figure~\ref{fig:swed_L_pred}(top right) shows the posterior predictive intervals for the $L$ functions, and while we do see some of the trends present
in the data, the fit is still not adequate. The problem now is that typical inter-event distances in the data are around $.5$ to $1$; however the 
presence of a few nearby points requires a small lower-bound $r_L$. 
The model then predicts many more nearby points than are observed in the data.

One remedy is to use some other distribution over the radii (e.g., a truncated Gaussian),
so that rather than driving the \matern parameters, occasional nearby events can be viewed as atypical.
Instead, we use our new extension with probabilistic thinning. This is easier to specify,
having only an interaction radius $R$ common to all \matern events, and a thinning probability $p$.
We place a unit mean $\text{Gamma}(1,1)$ 
and a flat $\text{Beta}(1,1)$ prior on these.
As before, we place a  $\text{Gamma}(1,1)$ prior on the Poisson intensity $\lambda$.


Figure~\ref{fig:swed_L_pred}(bottom left) shows the results from our MCMC sampler. 
The predicted  L-function values now mirror the empirical values much better, and in general we find probabilistic thinning strikes a 
better balance between simplicity and flexibility. The supplementary material also includes plots for another statistic,
the J-function \citep{vanLies96} agreeing with these conclusions.
Figure~\ref{fig:swedish_prob} plots the posteriors over the interaction radius, the thinning probability and the number of thinned events.
The first and last are much larger than the hardcore model, and coupled with a thinning probability of around $0.6$ confirm that the data are
indeed a repulsive process. 

We also estimated point process Fano factors by
dividing the space into $25$ $2 \times 2$ squares, and counting the number of events falling in each. Our estimate was then the standard 
deviation of these counts divided by the mean. A Poisson would have this equal to $1$, with values less than one suggesting regularity. The
rightmost subplot of figure~\ref{fig:swedish_prob} shows the empirical Fano factor (the red circle), as well as posterior predictive values 
(again obtained by calculating the Fano factors of synthetic datasets generated each MCMC iteration).
The plots confirm both that the data are repulsive and that the model captures this aspect of repulsiveness.

  \begin{figure}
  \centering
  \begin{minipage}[b]{0.245\linewidth}
  \centering
    \includegraphics[width=0.98\textwidth]{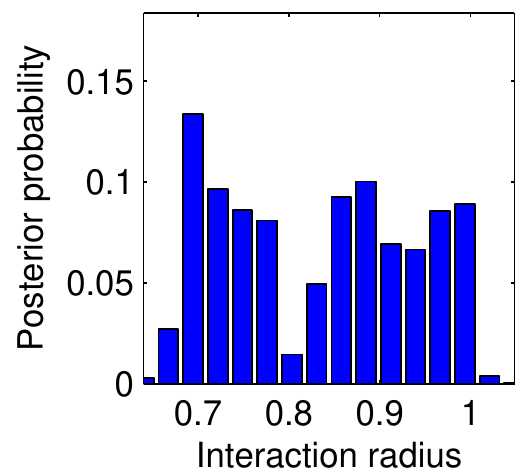} 
  \end{minipage}
  \begin{minipage}[b]{0.245\linewidth}
  \centering
    \includegraphics[width=0.98\textwidth]{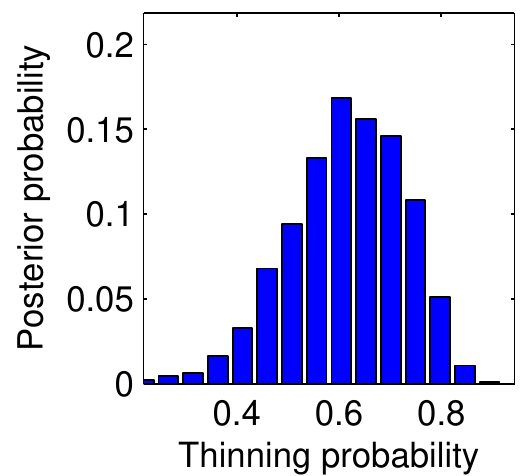} 
  \end{minipage}
  \begin{minipage}[b]{0.245\linewidth}
  \centering
    \includegraphics[width=0.98\textwidth]{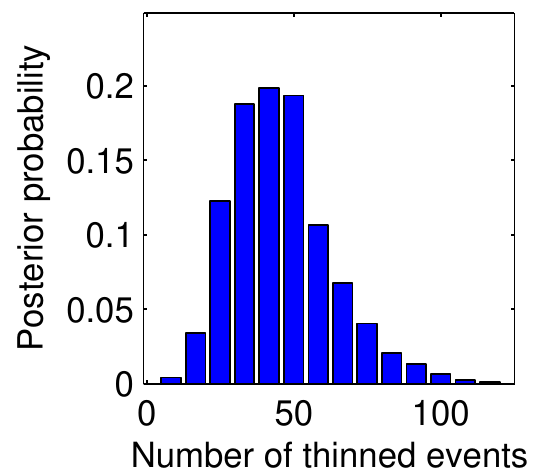} 
  \end{minipage}
  \begin{minipage}[b]{0.245\linewidth}
  \centering
    \includegraphics[width=0.98\textwidth]{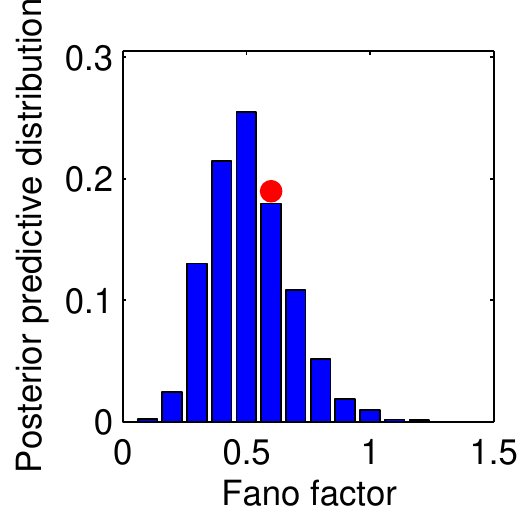} 
  \end{minipage}
  \caption[Posterior distributions of the homogeneous \matern type-III model with probabilistic thinning for the Swedish pine tree dataset]
      {Swedish tree dataset: posterior distributions of interaction radius $R$ (first), the thinning probability $p$ (second), and the number
      of thinned events (third). Rightmost is the empirical (circle) as well posterior predictive distribution over the Fano factor.
  }
  \label{fig:swedish_prob}
  \end{figure}
As a comparison, we also include a fit from a Gibbs-type process (a Strauss process, in particular) in the bottom right of 
figure~\ref{fig:swed_L_pred}. These
predictions of the L-function were produced from an MLE fit (using available routines from \texttt{spatstat}). For this dataset, the fit is comparable with
that of the \matern process with probabilistic thinning (though the L- and J-functions are only just contained in the $95\%$ prediction
band). Our next section considers a dataset where the Strauss process fares much worse.
\subsection{Nerve fiber data}
  \begin{figure}
  \centering
  \begin{minipage}[h]{0.22\linewidth}
    \centering \vspace{-.1in}
  \caption{Posterior distribution over thinning area times thinning probability for mild (left) and moderate (right) neuropathy}
  \label{fig:rep_inf}
  \end{minipage}
  \begin{minipage}[h]{0.3\linewidth}
  \centering
  \includegraphics[width=0.98\textwidth]{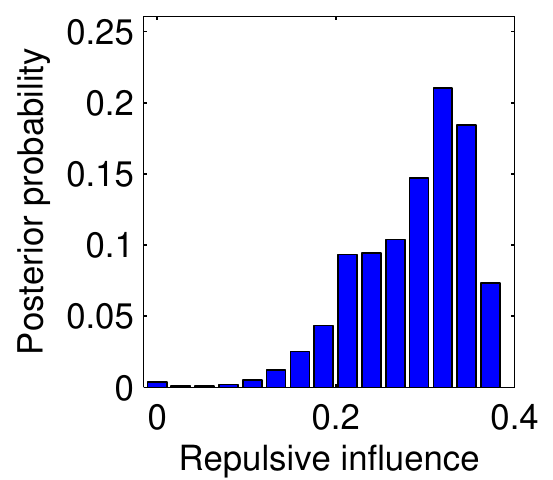}
  \end{minipage}
  \begin{minipage}[h]{0.3\linewidth}
  \centering
  \includegraphics[width=0.98\textwidth]{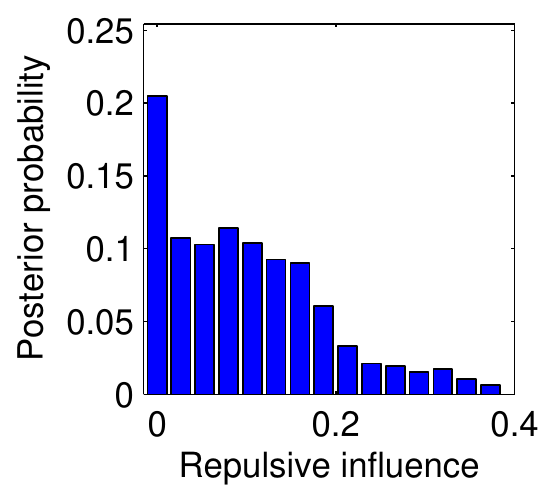}
  \end{minipage}
  \end{figure}
  \begin{figure}
  \begin{minipage}[h]{0.24\linewidth}
  \caption[Post Pred]{(Top row): Posterior predictive values of the L-functions for the \matern model with probabilistic thinning for mild (left) and 
  moderate (right) neuropathy. (Bottom row): Corresponding predictive values for Strauss process fits.}
  \label{fig:diab_Lfunc}
  \end{minipage}
  \centering
  \begin{minipage}[h]{0.74\linewidth}
  \begin{minipage}[h]{0.49\linewidth}
  \centering
  \includegraphics[width=0.98\textwidth]{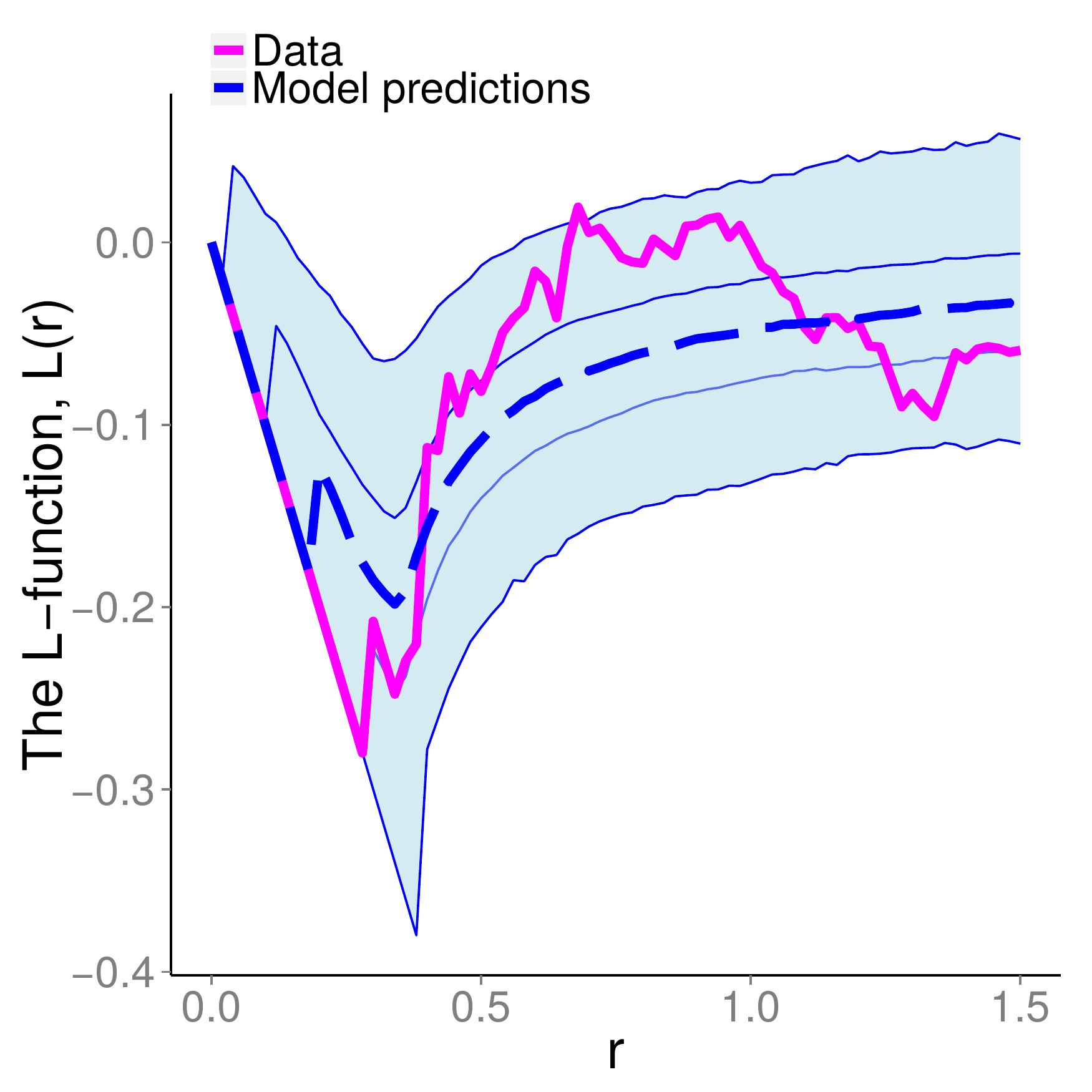}
  \end{minipage}
  \begin{minipage}[h]{0.49\linewidth}
  \centering
  \includegraphics[width=0.98\textwidth]{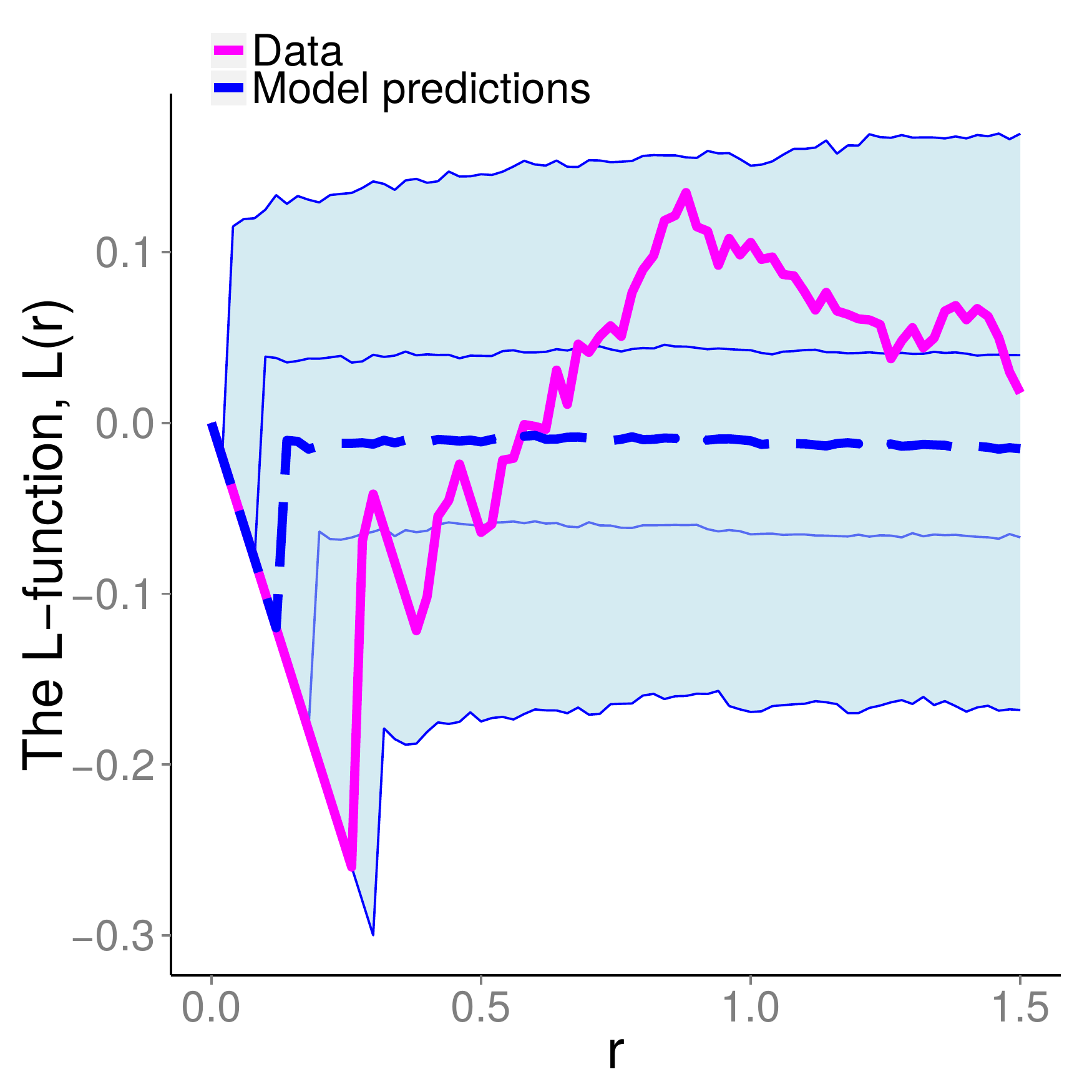}
  \end{minipage}
  \begin{minipage}[h]{0.49\linewidth}
  \centering
  \includegraphics[width=0.98\textwidth]{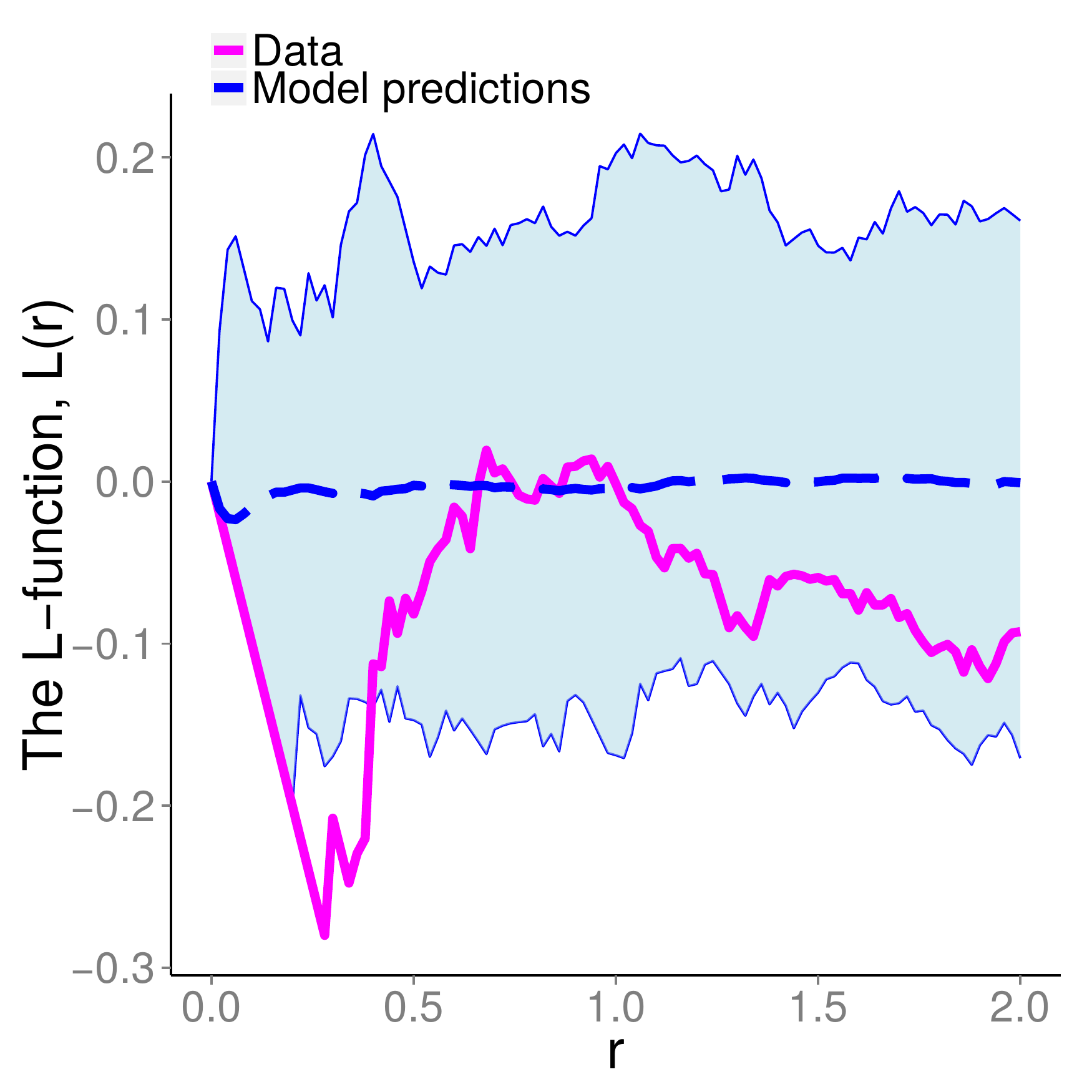}
  \end{minipage}
  \begin{minipage}[h]{0.49\linewidth}
  \centering
  \includegraphics[width=0.98\textwidth]{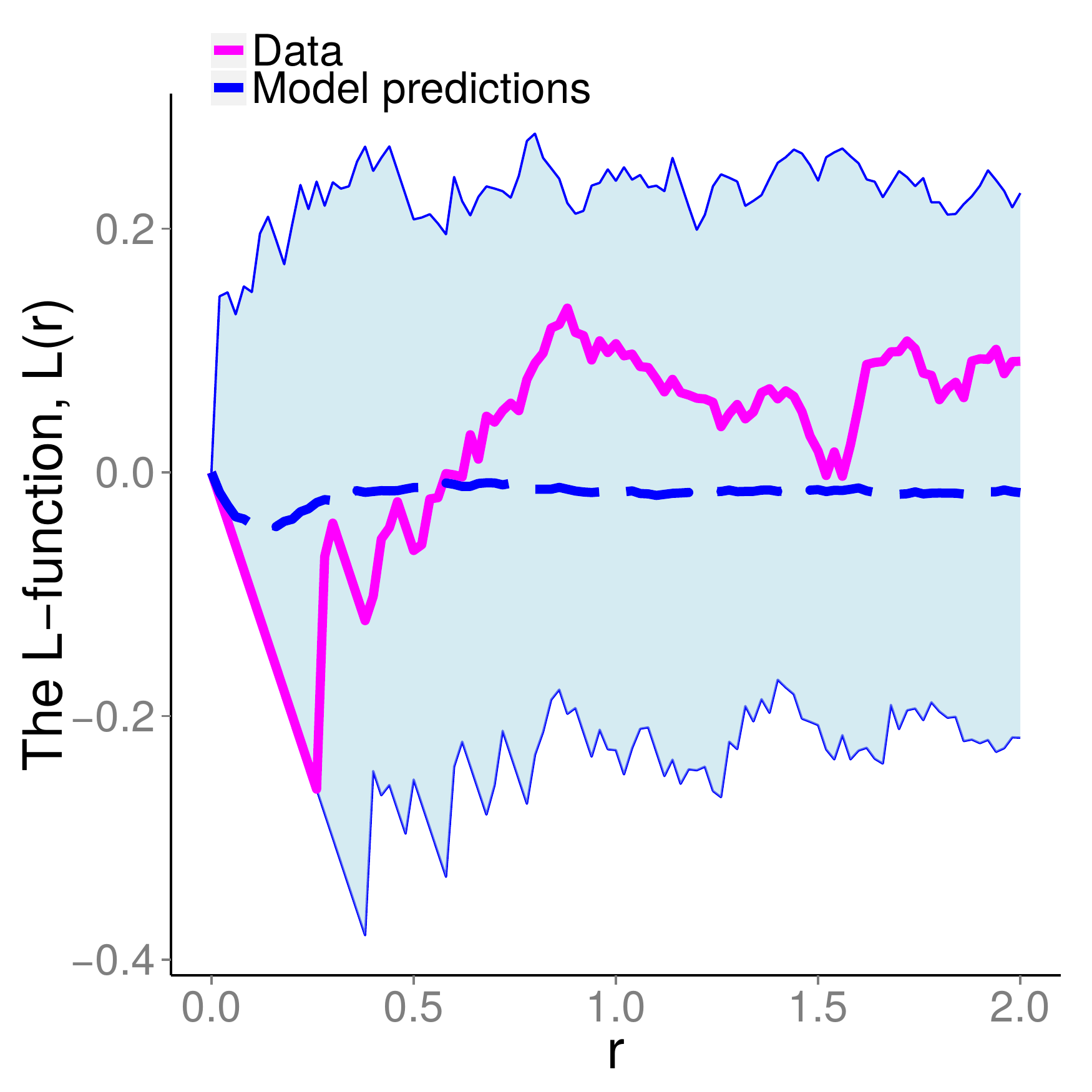}
  \end{minipage}
  \end{minipage}
  \end{figure}
For our next experiment, we consider a dataset of nerve fiber entry locations recorded from patients suffering from diabetes  \citep{WallSar11}.
Figure~\ref{fig:diab} plots these data for two patients at two stages of neuropathy, mild and moderate.
  As noted in \cite{WallSar11}, there is an increased clustering moving from mild to moderate, and this is confirmed by figure~\ref{fig:diab_Lfunc}, 
  which shows the empirical values for the L-function (solid magenta lines).
  Also included are posterior predictive values obtained from $10000$ posterior samples from the \matern model with probabilistic
  thinning. 
  The left plots for mild neuropathy show a repulsive effect (especially a dip in the L-function  around 
  the $0.5$ mark), and this is captured by the model. 
  By contrast, the moderate case exhibits repulsion and clustering at different scales, and the model settles with a Poisson process whose
  predictive intervals include such behavior.

To show the difference in the fits for the two conditions, at each MCMC iteration we calculate a statistic measuring repulsive influence: the thinning 
probability multiplied by the squared interaction radius. Figure~\ref{fig:rep_inf} plots this for the two cases, and we see that there is more
repulsion for the mild case. Another approach is to fix the thinning probability (to say $0.75$), and plot the posterior distributions over the
interaction radius. We get similar results, but have not reported them here. Distributions like these are important diagnostics towards
the automatic assessment of neuropathy.

  Figure~\ref{fig:diab_Lfunc} show the corresponding fits for a Strauss process. This is much
  worse, with the estimate essentially reducing to a Poisson process. 
  The failure of this model has partly to do with the fact that it is not as flexible as our \matern model
  with probabilistic thinning; another factor is the use of an MLE estimate as opposed to a fully Bayesian posterior.
  While it is possible to use a more complicated Gibbs-type process, inference (especially Bayesian inference) becomes
  correspondingly harder. Furthermore, as we demonstrate next, it is easy to incorporate nonstationarity into \matern
  processes; doing this for a Gibbs-type process is not at all straightforward. In the supplementary material, we include similar
  results for the J-function as well.

\section{Nonstationary \matern processes}  \label{sec:inhom_mat}
In many settings, it is useful to 
incorporate nonstationarity into repulsive models. 
For example, factors like soil fertility, rainfall and terrain can affect the density of trees at different locations. 
We will consider a dataset of locations of Greyhound bus stations: to allocate resources efficiently, these will be underdispersed.
At the same time, one expects larger intensities in urban areas than in more remote ones.
In such situations, it is important to account for nonstationarity while estimating repulsiveness

A simple approach is 
to allow the intensity of the primary Poisson $\lambda(s)$ to vary over~$\cS$. 
A flexible model of such a nonstationary intensity function is a transformed Gaussian process.
Let $\mathcal{K}(\cdot,\cdot)$ be the 
covariance kernel of a Gaussian process, $\hlambda$ a positive scale parameter, and 
$\sigma(x)=(1+\exp(-x))^{-1}$ the sigmoidal transformation. Then~$\lambda(\cdot)$ is a random function defined as
  \begin{align}
    \lambda(\cdot) = \hlambda \sigma(l(\cdot)), \quad
    l(\cdot) \sim \mathcal{GP}(0,\mathcal{K})\label{eq:sig_mod}
  \end{align}
  This model is closely related to the log-Gaussian Cox process, with a sigmoid link function instead of an exponential.
The sigmoid transformation serves two purposes: to ensure that the intensity $\lambda(s)$ is nonnegative, and
to provide a bound $\hlambda$ on the Poisson intensity. As shown in \cite{adams-murray-mackay-2009b}, such a bound makes drawing events from the 
primary process possible by a 
clever application of the thinning theorem. 
We introduced the thinning theorem in section~\ref{sec:inf_mat}; we state it formally below.
\begin{thrm}[Thinning theorem, \cite{Lewis1979}]
 Let $E$ be a sample from a Poisson process with intensity $\hlambda(s)$.
For some nonnegative function ${\lambda(s) \le \hlambda(s) \, \forall s \in \cS}$, assign each
point $e \in E$ to $F$ with probability $\frac{\lambda(e)}{\, \hlambda(e)}$.
 Then $F$ is a draw from a Poisson process with intensity $\lambda(s)$. \label{thrm:Thin}
\end{thrm}
Now, note from equation \eqref{eq:sig_mod} that ${\hlambda \ge \lambda(s)}$. Following the thinning theorem, 
one can obtain a sample from the rate $\lambda(s)$ inhomogeneous Poisson process by thinning a random sample $E$ from a rate $\hlambda$ homogeneous
Poisson process. Importantly, the thinning theorem requires us to instantiate the random intensity $\lambda(s)$ only on the elements of $E$,
avoiding any need to evaluate integrals of the random function $\lambda(s)$. We then have the
following retrospective sampling scheme: sample a homogeneous Poisson process with intensity $\hlambda$, and instantiate the Gaussian process $l(\cdot)$ on this sequence.
Keeping each element $e$ with probability ${\sigma(l(e))}$, 
we have an \emph{exact} sample from the inhomogeneous primary process. Call this $F$, and call the thinned events $\tF$. 
We can then use any of the \matern thinning schemes outlined previously to further thin $F$, resulting in an inhomogeneous repulsive process $G$.
Observe that there are now two stages of thinning, the first is an application of the 
Poisson thinning theorem to obtain the inhomogeneous primary process $F$ from the homogeneous Poisson process $E$, and the second, the \matern 
thinning to obtain $G$ from $F$. Algorithm~\ref{alg:inhom_mat} outlines the generative process of an inhomogeneous generalized \matern type-III process.
\begin{algorithm}
{Algorithm to sample an inhomogeneous \matern process on $\cS$} \label{alg:inhom_mat}
\begin{tabular}{p{1.4cm}p{12.2cm}}
\textbf{Input:}  & A Gaussian process prior $\mathcal{GP}(0, \mathcal{K})$ on the space $\cS$, a constant $\hat{\lambda}$ and \\
                 & the thinning kernel parameters $\theta$. \\
\textbf{Output:} & A sample $G$ from the nonstationary \matern type-III process. \\
\hline 
\end{tabular}
\begin{algorithmic}[1]
\State Sample $E$ from a homogeneous Poisson process with intensity $\hat{\lambda}$.
\State Instantiate the Gaussian process $l(\cdot)$ on these points. Call vector $l_E$.
\State 
       Keep a point $e \in E$ with probability $\sigma(l(e))$, otherwise thin it. The surviving points form the primary process $F$.
\State Assign $F$ a set of random birth times $T^F$, independently and uniformly on $[0,1]$.
\State Proceed through elements of $F$ in order of birth. At each element, evaluate the shadow of the previous surviving elements,
       and keep it or thin it as appropriate.
\State The surviving points $G$ form the inhomogeneous \matern type-III point process.
\end{algorithmic}
\end{algorithm}

As in section~\ref{sec:inf_mat}, we place priors on $\hlambda$ as well as $\theta$. We also place hyperpriors on the hyperparameters of the GP covariance kernel. 

\subsection{Inference for the inhomogeneous \matern type-III process}


Proceeding as in section~\ref{sec:matern_pdf},
it follows that the events $\tG^+$ thinned by the repulsive kernel are conditionally distributed as an inhomogeneous Poisson process with 
intensity $\lambda(s) \mathscr{H}(s,t;G^+)$ (with the $\lambda$ from corollary~\ref{prop:mat_post} replaced by $\lambda(s)$). From the thinning theorem, simulating from such a process is a simple matter of thinning events of a homogeneous,
rate $\hlambda$ Poisson process, exactly as outlined in 
the previous section. Having reconstructed the inhomogeneous primary Poisson process, the update rules for the birth times $T_G$, and the thinning kernel parameters 
$\theta$ are identical to the homogeneous case. 

The only new idea 
involves updating the random function $\lambda(s)$ (more precisely, the latent GP, $l(s)$, and the scaling factor, $\hlambda$). 
To do this, we first instantiate $\tF$, the events of $E$ thinned in constructing the inhomogeneous primary process $F$.
For this step, \citet{adams-murray-mackay-2009b} constructed a Markov transition kernel, involving a set of birth-death moves that updated the number of thinned 
events, as
well as a sequence of moves that perturbed the locations of the thinned events. This kernel was set up to have as equilibrium distribution the required
posterior over the thinned events. 
Instead, similar in spirit to our idea of jointly simulating the thinned \matern events, it is possible to produce a conditionally independent sample of $\tF$.
Instead of a number of local moves that perturb the current setting of $\tF$ in the Markov chain, we can discard the old thinned
events, and jointly produce a new sample given the rest of the variables. 
The required joint distribution is given by the following corollary of Theorem~\ref{thrm:Thin}:
\begin{coro} \label{prop:thin_post} Let $F$ be a sample from a Poisson process with intensity $\lambda(s)$, produced
by thinning a sample $E$ from a Poisson process with intensity $\hlambda$. Call the thinned events $\tF$. Then given $F$, $\tF$ is a 
Poisson process with intensity $\hlambda - \lambda(s)$.
\end{coro}
\begin{proof}
A direct approach uses the Poisson densities defined in Theorem~\ref{thrm:poiss_density}. More intuitive is the following: by symmetry, the
thinning theorem implies that the thinned events $\tF$ are distributed as a Poisson process with intensity $\hlambda - \lambda(s)$. By the
independence property of the Poisson process, $F$ and $\tF$ are independent, and the result follows.
\end{proof}
Applying this result to simulate the thinned events $\tF$ is a straightforward application of the thinning theorem: 
sample from a homogeneous Poisson process with intensity $\hlambda$, conditionally instantiate the function $\lambda(\cdot)$ on these points (given its values on $F$), 
and keep element $\tilde{f}$ with probability $1 - \sigma(\tilde{f})$. 

Having reconstructed the sequence $E = F \cup \tF$, we can update the values of the GP instantiated on $E$. 
Recall that an element $e \in E$ is assigned to $F$ with probability~$\sigma(l(s))$, otherwise it is
assigned to $\tF$. Thus updating the GP reduces to updating the latent GP in a classification problem with a sigmoidal link function.
There are a variety of sampling algorithms to simulate such a latent Gaussian process, in our experiments we used the elliptical slice sampling from 
\cite{murray2010}. Finally, as the number of elements of $E$ is Poisson distributed with rate $\hlambda$, a Gamma prior results in a Gamma posterior as 
in section~\ref{sec:Poiss_int_inf}.
We describe our overall sampler in detail in algorithm~\ref{alg:mat_np_inf}.
\algblock[sampleG]{StartG}{EndG}
\algrenewtext{StartG}{\textbf{Sample the \matern thinned events $\tG^+$:}}
\algrenewtext{EndG}{}

\algblock[sampleF]{StartF}{EndF}
\algrenewtext{StartF}{\textbf{Sample the Poisson thinned events $\tF$:}}
\algrenewtext{EndF}{}

\algblock[sampleMisc]{StartMisc}{EndMisc}
\algrenewtext{StartMisc}{\textbf{Sample the \matern birth-times:}}
\algrenewtext{EndMisc}{}

\algblock[sampleParam]{StartParam}{EndParam}
\algrenewtext{StartParam}{\textbf{Sample the parameters $\hlambda$ and $\theta$:}}
\algrenewtext{EndParam}{}

\begin{algorithm}
{MCMC update for inhomogeneous \matern type-III process on $\cS$}
\begin{tabular}{p{1.4cm}p{12.2cm}}
\textbf{Input:}  & \matern events with birth times $G^+ \equiv (G,T^G)$, \\
    & Thinned primary events $\tG^+ \equiv (\tG,T^{\tG})$ and thinned Poisson events $\tilde{F}$ \\
    & A GP realization $l_E$ on $E \equiv G \cup \tG \cup \tF$\\
    & Parameters $\hlambda$ and $\theta$  \\
\textbf{Output:} & New realizations $T^G_{new}$, $\tG_{new}^+$, $\tF_{new}$, \\
    &  A new instantiation of the GP on $G \cup \tG_{new} \cup \tF_{new}$. \\
    &  New values of $\hlambda$ and $\theta$ \\
\hline 
\end{tabular}
\begin{algorithmic}[1]
\StartG
\State  \begin{minipage}{\textwidth} Discard the old \matern- thinned event locations $\tG^+$.\end{minipage}
\State Sample events $A^+ \equiv (A, T^A)$ from a rate $\hlambda$ Poisson process on $\cS \times \cT$. 
\State Sample $l_A | l_E$ (conditionally from a multivariate normal).
\State Keep a point $a \in A$ with probability $\sigma(l(a)) \mathscr{H}(a;G^+)$, otherwise thin it. 
\State The surviving points form the new \matern thinned events $\tG^+_{new}$.
\State Discard GP evaluations on old \matern events, and add the new ones to $l_E$.
\EndG
\StartF
\State Define $E_{new} \equiv G \cup \tG_{new} \cup \tF$
\State  \begin{minipage}{\textwidth} Discard the old thinned primary Poisson events $\tF^+$.\end{minipage}
\State Sample events $B^+ \equiv (B, T^B)$ from a rate $\hlambda$ Poisson process on $\cS \times \cT$. 
\State Sample $l_B | l_{E_{new}}$ (conditionally from a multivariate normal).
\State  \begin{minipage}{\textwidth} Keep a point $b \in B$ with probability $1 - \sigma(l(b))$, otherwise thin it. 
       The surviving points form the new Poisson thinned events $\tF_{new}$.\end{minipage}
\State Define $E_{new} \equiv G \cup \tG_{new} \cup \tF_{new}$
\EndF
\StartMisc
\State  \begin{minipage}{\textwidth} For each \matern observation $g$, resample its birth time $T^G_g$ conditioned on all other variables.\end{minipage}
\EndMisc
\State \label{step:gp_inf}\!\textbf{Sample the GP values $l_{E_{new}}$:} 
\State \hspace{.2in} We used elliptical slice sampling \citep{murray2010}.
\State
\StartParam
\State  \begin{minipage}{\textwidth} Sample $\hlambda$ and $\theta$ given the remaining variables.\end{minipage}
\EndParam
\end{algorithmic}
\label{alg:mat_np_inf}
\end{algorithm}

\subsection{Experiments}

We consider a dataset of the locations of $79$ Greyhound bus stations in three states in the southeast United States (North Carolina, South Carolina and 
Georgia)\footnote{Obtained from \texttt{http://www.poi-factory.com/}.}. Figure~\ref{fig:grey_nc_in}(left) plots these locations along with the inferred intensity function
from our nonstationary \matern process. We modeled the intensity function $\lambda(\cdot)$ using a sigmoidally-transformed Gaussian process 
with zero-mean and a squared-exponential kernel. We placed log-normal hyperpriors on the 
 GP hyperparameters ((shape, scale) as $(1,1))$, and a Gamma$(10,10)$ prior on the scale parameter $\hlambda$ (we allowed larger variance since
 there are two levels of thinning).
 To ease comparison across two cases, we fixed the thinning probability to $0.75$, and placed a Gamma$(1,1)$ prior on the interaction radius $R$ (which
 we expect to be less than one degree of latitude/longitude).
The contours in figure~\ref{fig:grey_nc_in}(left) correspond to the posterior mean of $\lambda(\cdot)$, obtained from $10000$ MCMC samples 
following algorithm~\ref{alg:mat_np_inf}.
 We used elliptical slice sampling \citep{murray2010} to resample the GP values 
(step~\ref{step:gp_inf} in algorithm~\ref{alg:mat_np_inf}). and the GP hyperparameters were resampled by slice-sampling \citep{murray2010b}. 

We see a strong nonstationarity, driven mainly by the Atlantic Ocean to the lower right, and the absence of recordings to the top left (in Tennessee).
Not surprisingly, the intensity function takes large values over the three states of interest. The right subplot shows the same analysis, now
restricted to North Carolina (with observations in the other states discarded). Here we see more clearly a fine structure in the intensity
function with two peaks, one in the research triangle area (Raleigh-Durham-Chapel Hill) to the west, and the other corresponding to the Charlotte
metropolitan area. 
  \begin{figure}
  \centering
  \begin{minipage}[h]{0.45\linewidth}
  \centering
    \includegraphics[width=0.99\textwidth, angle=0]{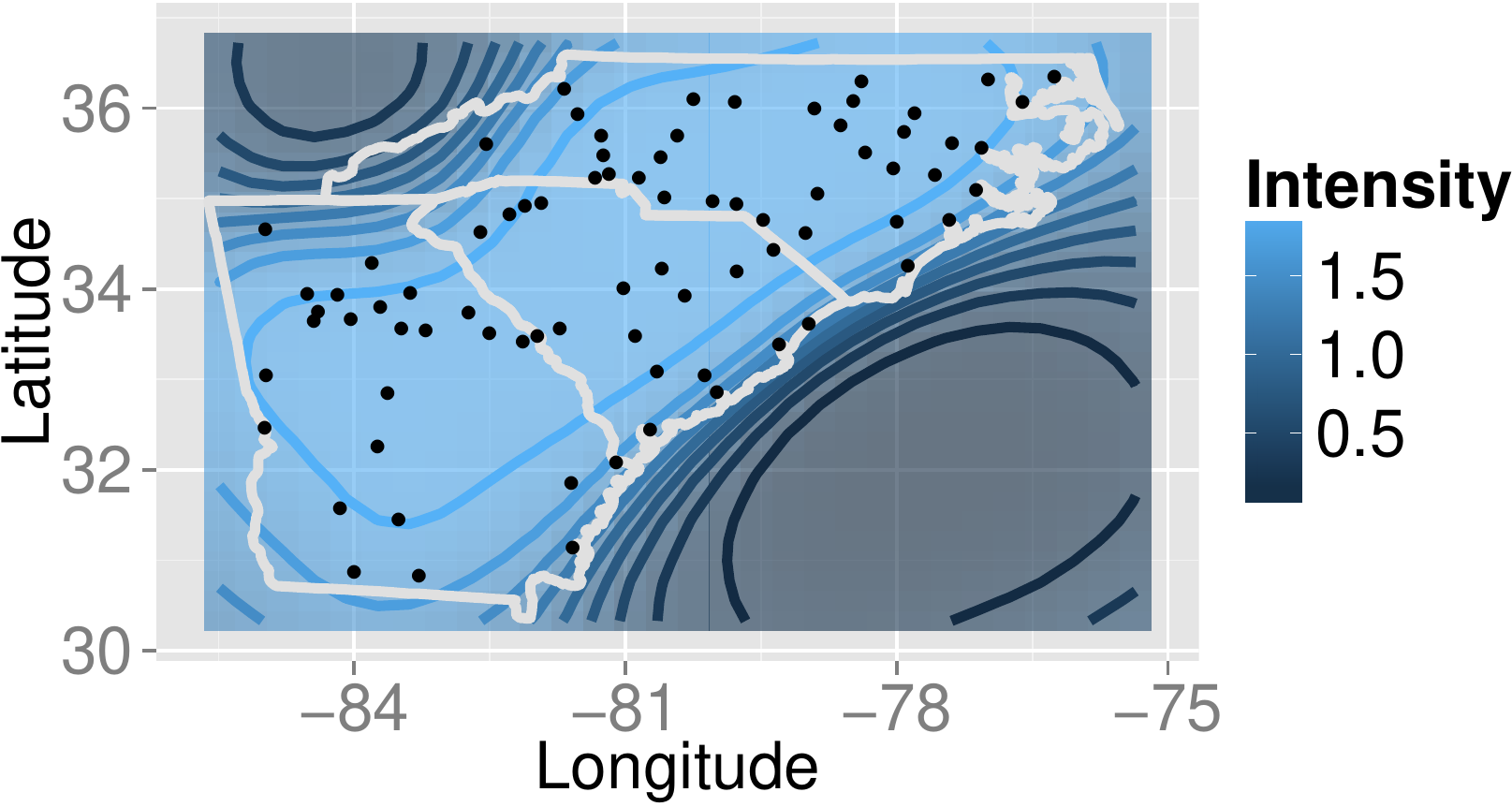}
  \end{minipage}
  \begin{minipage}[h]{0.48\linewidth}
  \centering
    \includegraphics[width=0.99\textwidth, angle=0]{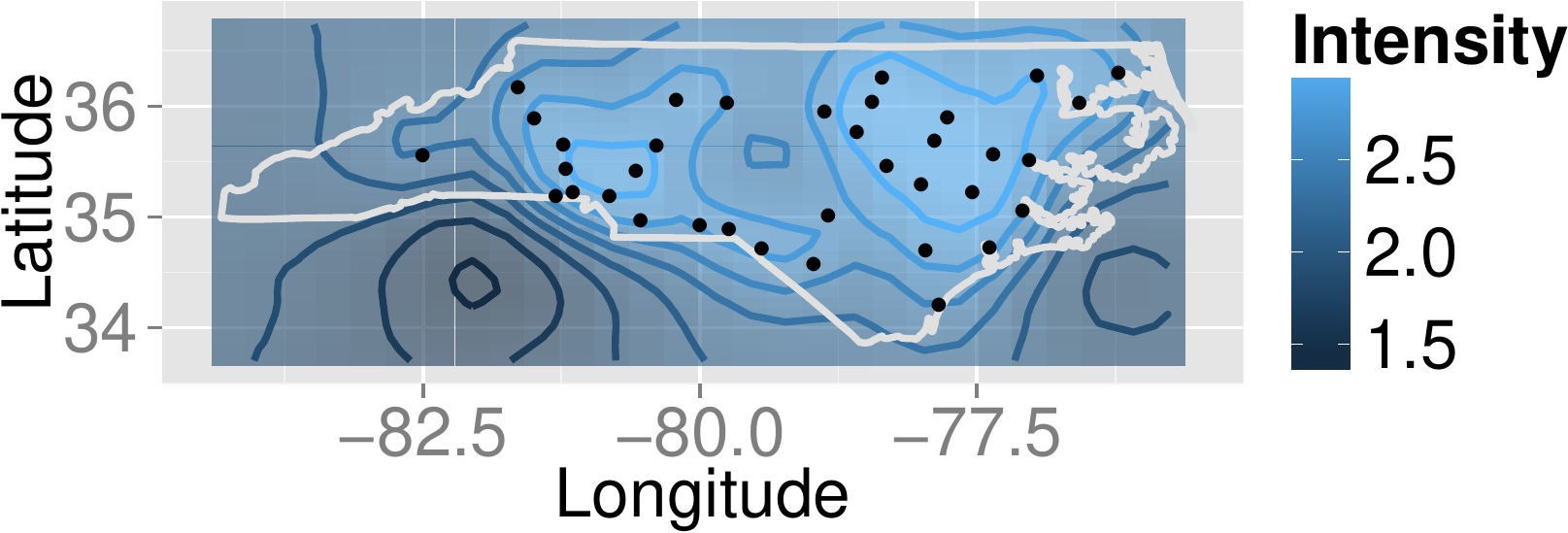}
    \vspace{-.3in}
    \caption{Posterior mean of the intensity for the greyhound dataset: (left) all three states (right) North Carolina}
  \label{fig:grey_nc_in}
  \end{minipage}
  \end{figure}

  \begin{figure}
  \centering
  \begin{minipage}[h]{0.23\linewidth}
  \caption{Posterior over the thinning radius for (left) all three states, and (right) North Carolina}
  \label{fig:grey_nc_in_rad}
  \end{minipage}
  \begin{minipage}[h]{0.28\linewidth}
    \centering
    \includegraphics[width=0.98\textwidth, angle=0]{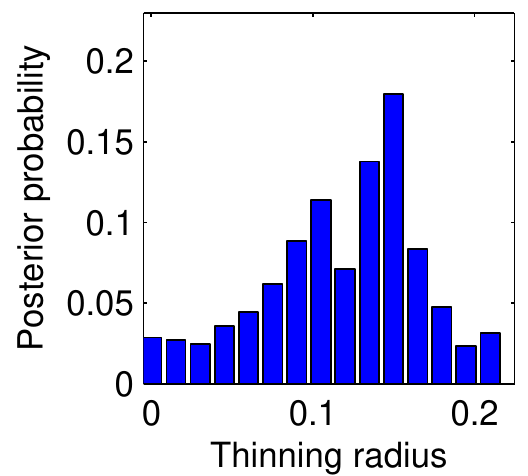}
  \end{minipage}
  \begin{minipage}[h]{0.28\linewidth}
  \centering
    \includegraphics[width=0.98\textwidth, angle=0]{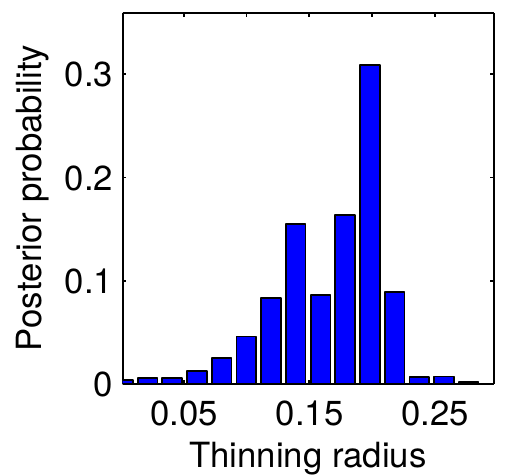}
  \end{minipage}
  \end{figure}

  \begin{figure}
  \centering
  \begin{minipage}[h]{0.28\linewidth}
  \caption{Inhomogeneous L-function for the Greyhound dataset: (left) posterior predictive values for nonstationary \matern with probabilistic thinning, and (right)
  fit of an inhomogeneous Poisson process} 
  \label{fig:grey_ncscga_lfunc}
  \end{minipage}
  \begin{minipage}[h]{0.32\linewidth}
    \centering
    \includegraphics[width=0.98\textwidth, angle=0]{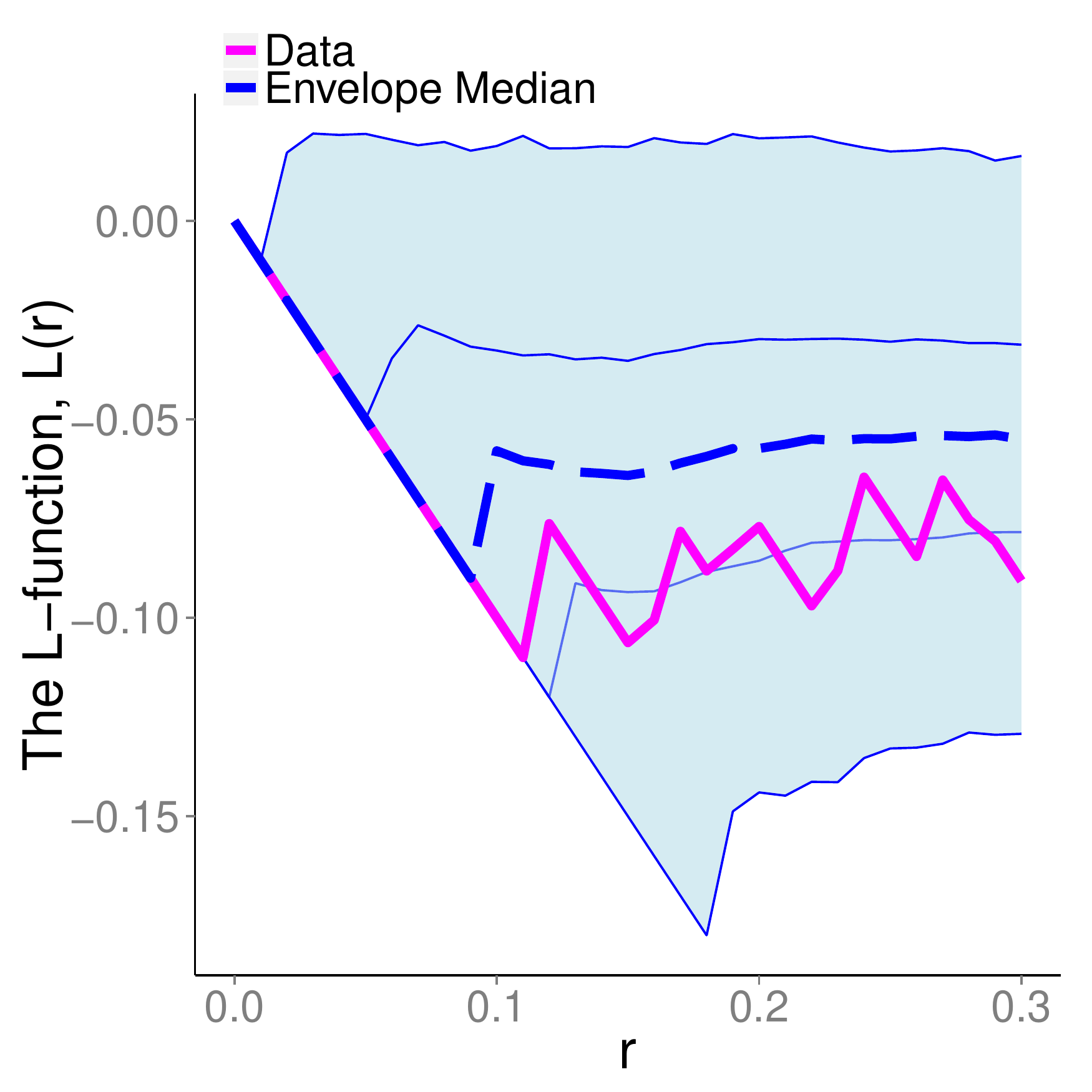}
  \end{minipage}
  \begin{minipage}[h]{0.32\linewidth}
  \centering
    \includegraphics[width=0.98\textwidth, angle=0]{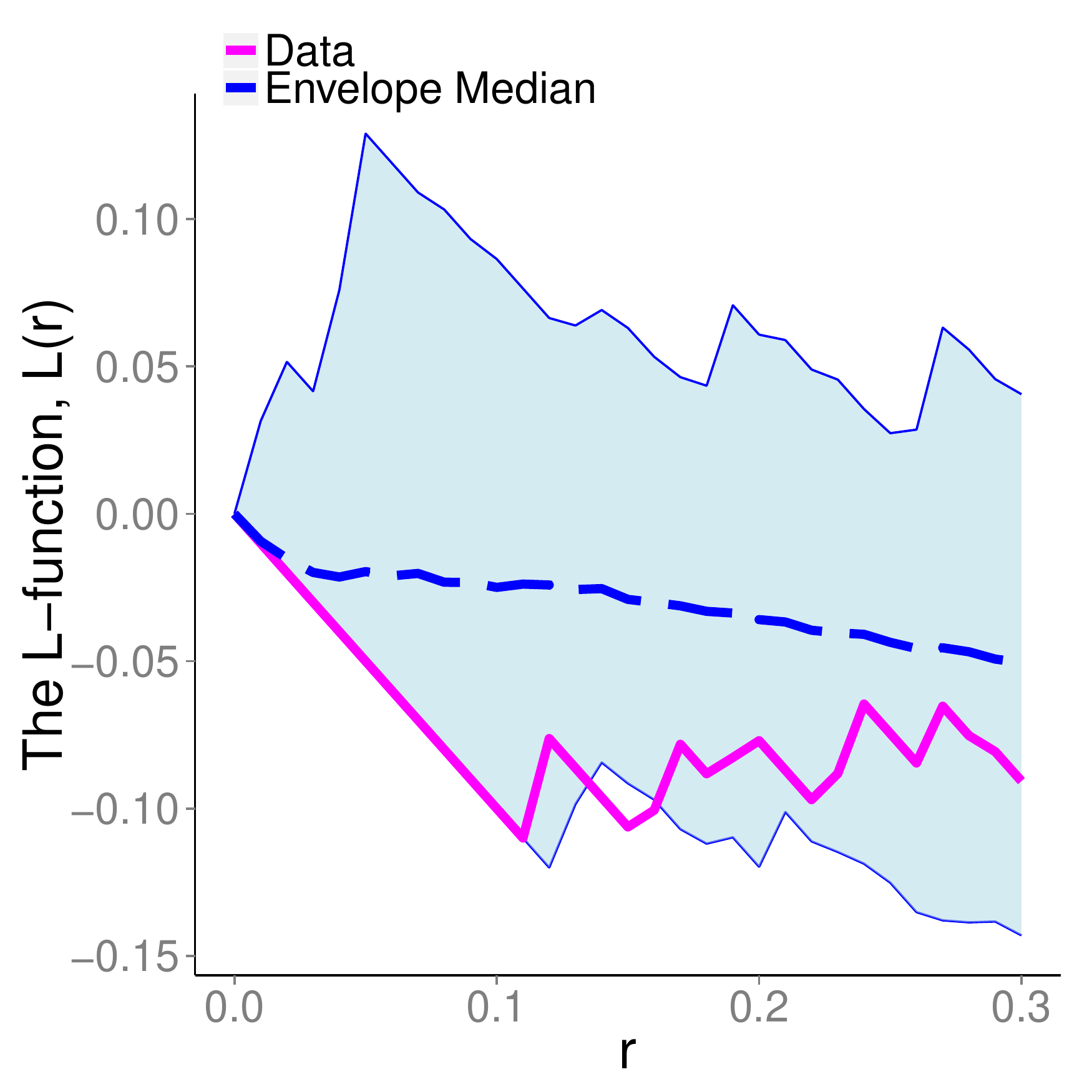}
  \end{minipage}
  \end{figure}
Figure~\ref{fig:grey_nc_in_rad} plots the posterior distributions over the thinning radius $R$ for both cases (recall that
since we used a fixed thinning probability of $0.75$, $R$ is a direct measure of the repulsive effect). 
In both cases, we see this distribution is peaked around $0.15$ to $0.2$ degrees, with the North Carolina dataset having a slightly clearer
repulsive effect. 
For the grouped analysis, we required all three states to share the same unknown $R$ while having their own nonstationarity: more generally, one can 
use a hierarchical model over $R$ to shrink or cluster the radii across different states. The resulting posterior over $R$ can help understand 
variations in pricing, reliability and delays, as well as assist towards future developmental work~\citep{Sahin2007}.

Figure~\ref{fig:grey_ncscga_lfunc} compares fits for a nonstationary \matern process and a Poisson process, using the inhomogeneous
L-function \citep{Baddely2000}. 
We see that the Poisson process (to the right) fails
to capture the empirical values over distances up to around $0.2$ degrees. This agrees with  figure~\ref{fig:grey_nc_in_rad} that repulsion occurs over 
these distances. The \matern process (to the left) provides a much better fit. 
We also include plots for the the nonstationary J-function in the supplementary material. This did not indicate much deviation from Poisson, 
though as \cite{Baddely2000} point out, this does not imply that the point process is Poisson. In any event, both models did near identical jobs 
capturing the J-function.

\section{Discussion}
We described a Bayesian framework for modeling repulsive interactions between events of a point process based on the \matern type-III point process. Such a framework allows flexible and intuitive repulsive effects between events, with parameters that are interpretable and realistic.
We developed an efficient MCMC sampling algorithm for posterior inference in these models,
and applied our ideas to a dataset of tree locations, and two datasets each of nerve-fiber locations and bus-station locations.

There are a number of interesting directions worth following.
While we only considered events in a $2$-dimensional space,
it is easy to generalize to higher dimensions to model, say, the distribution of galaxies in space or features in some 
feature space.
As with all repulsive point processes, very high densities lead to inefficiency (in our case because of very high primary process intensities).
It is important to better understand the theoretical and computational limitations of our model in this regime.
While we assumed the \matern events $G$ were observed perfectly, there is often noise into this observation process. In this case, given
the observed point process $G_{obs}$, we have to instantiate the latent \matern process $G$. 
Simulating the locations of the events in $G$ would require incremental updates, and if we allow for missing or extra events, we would need a 
birth-death sampler as well. A direction for future study is to see how these steps can be performed efficiently.
Having instantiated $G$, all other variables can be simulated as outlined in this paper. 
A related question concerns whether our ideas can be extended to develop efficient samplers for \matern type I and II processes as well.

Our models assumed a homogeneity in the repulsive properties of the \matern events. An interesting extension is to allow, say, the interaction radius or 
the thinning probability to vary spatially (rather than the Poisson intensity $\lambda(\cdot)$). 
Similarly, one might assume a clustering of these repulsive parameters; 
this is useful in situations where the \matern observations represent cells of different kinds. 
Finally, it is of interest to apply our ideas to hierarchical models that don't necessarily represent point pattern data, for instance to encourage diversity between cluster parameters in a mixture model.

\section{Acknowledgements}
This work was partially supported by the DARPA MSEE project and by NSF IIS-1421780. We thank Aila S{\"a}rkk{\"a} for providing us with the nerve fiber data. The Greyhound dataset
was obtained from \texttt{http://www.poi-factory.com/}.
\appendix
\section{Appendix}

\setcounter{defn}{1}

\begin{thrm}
Let $G^+ = (G, T^G)$ be a sample from a generalized \matern type-III process, augmented with the birth 
times. Let $\lambda$ be its intensity, and $\mathscr{H}_{\theta}(s,t;G^+)$ be its shadow following the appropriate thinning scheme. 
Then, its density w.r.t.\  ${\mu}^{\cup}$ is 
\begin{align}
  p(G^+ \given \lambda, \theta) &=\exp\left(-\lambda \int_{\cS \times \cT}\left( 1 - \mathscr{H}_{\theta}(s,t;G^+)\right)\mu(\dif s \,\dif t)\right)
              \lambda^{|G^+|}  \nonumber \\
              & \quad \times  \prod_{g^+ \in G^+} \left( 1 - \mathscr{H}_{\theta}(g^+;G^+) \right).
\label{eq:mat_marg_prob_app}
\end{align}

\end{thrm}
\begin{proof}
 As described in Section~\ref{sec:matern_pdf},  $G^+= (G, T^G)$ is a sequence in the union space $(\cS\times\cT)^{\cup}$. Its elements are ordered by the last dimension, so that $T^G$ is an
  increasing sequence. Let $|G^+|$, the length of $G^+$, be $k$. $G^+$ is obtained by thinning $F^+=(F,T^F)$, a sample from a homogeneous Poisson 
process with intensity $\lambda$.
Let the size of $F^+$ be $n \ge k$, and call the thinned points $\tG^+$. 
From Theorem~\ref{thrm:poiss_density}, the density of $F^+$ w.r.t.\ the measure ${\mu}^{\cup}$ is 
\begin{align}
 p(F^+) &= \exp\left(-\lambda \mu(\cS \times \cT)\right) \lambda^n. \label{eq:mat_prior_prob_app}
\end{align}
 Recall the definition of $\mathscr{H}_{\theta}(s,t; G^+)$, the shadow of $G^+$
(with $K_{\theta}$ the thinning kernel):
\begin{align}
  \mathscr{H}_{\theta}(s,t;G^+) &= 1 - \prod_{(s^*,t^*) \in G^+} (1 - I(t > t^*) K_{\theta}(s^*, s)) 
\end{align}
 We traverse sequentially through $F^+$, assigning element $i$ to the \matern process $G^+$ or the thinned set $\tG^+$ with probability determined
 by the shadow $\mathscr{H}_{\theta}(\cdot;G^+)$, and 
\begin{align}
  p(G^+, \tG^+ | F^+) = \prod_{(s,t) \in \tG^+} \mathscr{H}_{\theta}(s,t;G^+)  \prod_{(s,t) \in G^+} (1- \mathscr{H}_{\theta}(s,t;G^+)).   \label{eq:mat_marg_app}
\end{align}
In our notation above, the shadow $\mathscr{H}_{\theta}$ thins or keeps points to form $G^+=(G,T^G)$ and $\tG^+=(\tG, T^{\tG})$, but $\mathscr{H}_{\theta}$ also depends on 
$G^+$. There is no circularity however, since the shadow of a later point cannot affect an earlier point. The joint probability is 
\begin{align}
   p(G^+,\tG^+,&F^+) = p(G^+, F^+) = p(G^+, \tG^+) = \nonumber \\  
   \exp(-\lambda &\mu(\cS \times \cT)) \lambda^{n} 
              \prod_{(s,t) \in \tG^+} \mathscr{H}_{\theta}(s,t;G^+)  \prod_{(s,t) \in G^+} (1- \mathscr{H}_{\theta}(s,t;G^+)).   \label{eq:mat_joint_app}
            \intertext{Integrating out the $\cS$-locations of $n-k$ thinned elements } 
   p\left(G^+, T^{\tG} \right)  &=  \exp(-\lambda \mu(\cS \times \cT))
             \lambda^{k}  \prod_{(s,t) \in G^+} (1- \mathscr{H}_{\theta}(s,t;G^+)) \nonumber \\
             & \quad \prod_{t^{\tilde{g}} \in T^{\tG}} {\left( \lambda \int_{\cS } \mathscr{H}_{\theta}(s,t^{\tilde{g}};G^+)\mu(\dif s) \right)}. \label{eq:mat_l} 
\end{align}
{
We now integrate out the values of $T^{\tG}$, noting it is an ordered sequence of $(n-k)$ elements in $[0,1]$. We are left with 
$p(G^+, |T^{\tilde{G}}|=n-k )$, the joint probability of a sequence $G^+$ and that there are $n-k$ thinned events:
\begin{align}
  p\left(G^+, |T^{\tilde{G}}|=n-k \right)  &=  \exp(-\lambda \mu(\cS \times \cT))
             \lambda^{k}  \prod_{(s,t) \in G^+} (1- \mathscr{H}_{\theta}(s,t;G^+)) \nonumber \\
          & \quad \frac{\left( \lambda \int_{\cS \times \cT} \mathscr{H}_{\theta}(s,t;G^+)\mu(\dif s\ \dif t) \right)^{n-k}}{(n-k)!}.  
\end{align}
Finally, summing over values of $|T^{\tG}|$,}
\begin{align}
 p\left(G^+ \right)  
           &=   \exp\left(-\lambda  \int_{\cS \times \cT} \left(1 -  \mathscr{H}_{\theta}(s,t;G^+) \right) \mu(\dif s\ \dif t)  \right)
                 \prod_{(s,t) \in G^+} \lambda (1- \mathscr{H_{\theta}}(s,t;G^+)).
\end{align}


This is what we set out to prove. \hfill ${}_\blacksquare$
\end{proof}
\bibliographystyle{apalike}
\bibliography{refvr}

\begin{thebibliography}{}

\bibitem[Adams, 2009]{adamsthesis}
Adams, R.~P. (2009).
\newblock {\em Kernel methods for nonparametric Bayesian inference of
  probability densities and point processes}.
\newblock PhD thesis, University of Cambridge.

\bibitem[Adams et~al., 2009]{adams-murray-mackay-2009b}
Adams, R.~P., Murray, I., and MacKay, D. J.~C. (2009).
\newblock Tractable nonparametric {B}ayesian inference in {P}oisson processes
  with {G}aussian process intensities.
\newblock In {\em Proceedings of the 26th International Conference on Machine
  Learning (ICML)}.

\bibitem[Affandi et~al., 2014]{affandi2014learning}
Affandi, R.~H., Fox, E., Adams, R.~P., and Taskar, B. (2014).
\newblock Learning the parameters of determinantal point process kernels.
\newblock In {\em Proceedings of the 31st International Conference on Machine
  Learning}, pages 1224--1232.

\bibitem[Andrieu and Roberts, 2009]{AndRob10}
Andrieu, C. and Roberts, G.~O. (2009).
\newblock {The pseudo-marginal approach for efficient Monte Carlo
  computations}.
\newblock {\em The Annals of Statistics}, 37(2):697--725.

\bibitem[Baddeley and Turner, 2005]{spatstat}
Baddeley, A. and Turner, R. (2005).
\newblock Spatstat: an {R} package for analyzing spatial point patterns.
\newblock {\em Journal of Statistical Software}, 12(6):1--42.
\newblock {ISSN} 1548-7660.

\bibitem[Baddeley et~al., 2000]{Baddely2000}
Baddeley, A.~J., M{\o}ller, J., and Waagepetersen, R. (2000).
\newblock Non- and semi-parametric estimation of interaction in inhomogeneous
  point patterns.
\newblock {\em Statistica Neerlandica}, 54(3):329--350.

\bibitem[Besag, 1977]{Besag77}
Besag, J. (1977).
\newblock Contribution to the discussion of {D}r {R}ipley's paper.
\newblock {\em Journal of the Royal Statistical Society}, 29:193--195.

\bibitem[Brown et~al., 2004]{Brown2004a}
Brown, E.~N., Barbieri, R., Eden, U.~T., and Frank, L.~M. (2004).
\newblock {Likelihood methods for neural spike train data analysis}.
\newblock In Feng, J., editor, {\em Computational Neuroscience: A Comprehensive
  Approach}, volume~7, chapter~9, pages 253--286. CRC Press.

\bibitem[Daley and {Vere-Jones}, 2008]{DalVer2008a}
Daley, D.~J. and {Vere-Jones}, D. (2008).
\newblock {\em An Introduction to the Theory of Point Processes}.
\newblock Springer.

\bibitem[Hill, 1973]{Hill1973}
Hill, M.~O. (1973).
\newblock The intensity of spatial pattern in plant communities.
\newblock {\em Journal of Ecology}, 61(1):pp. 225--235.

\bibitem[Hough et~al., 2006]{hough_dpp}
Hough, J.~B., Krishnapur, M., Peres, Y., Vir{\'a}g, B., et~al. (2006).
\newblock Determinantal processes and independence.
\newblock {\em Probability Surveys}, 3(206-229):9.

\bibitem[Huber and Wolpert, 2009]{Hube:Wolp:2009}
Huber, M.~L. and Wolpert, R.~L. (2009).
\newblock Likelihood-based inference for {M}at{\'e}rn type~{III} repulsive
  point processes.
\newblock {\em Advances in Applied Probability}, 41(4).
\newblock 958--977.

\bibitem[Kendall, 1939]{kendall39}
Kendall, M.~G. (1939).
\newblock The geographical distribution of crop productivity in {E}ngland.
\newblock {\em Journal of the Royal Statistical Society}, 102(1):21--62.

\bibitem[Knox, 2004]{Knox04}
Knox, G. (2004).
\newblock {Epidemiology of Childhood Leukemia in Northumberland and Durham}.
\newblock {\em The Challenge of Epidemiology: Issues and Selected Readings},
  1(1):384--392.

\bibitem[Lewis and Shedler, 1979]{Lewis1979}
Lewis, P. A.~W. and Shedler, G.~S. (1979).
\newblock Simulation of nonhomogeneous {P}oisson processes with degree-two
  exponential polynomial rate function.
\newblock {\em Operations Research}, 27(5):1026--1040.

\bibitem[Mat{\'e}rn, 1960]{Matern60}
Mat{\'e}rn, B. (1960).
\newblock Spatial variation.
\newblock {\em Meddelanden fr{\aa}n Statens Skogsforskningsinstitut (Reports of
  the Forest Research Institute of Sweden)}, 49(5).

\bibitem[Mat{\'e}rn, 1986]{Matern86}
Mat{\'e}rn, B. (1986).
\newblock {\em {Spatial Variation}}.
\newblock Lecture notes in Statistics. Springer-Verlag, second edition.

\bibitem[Mateu and Montes, 2001]{Mateu2001}
Mateu, J. and Montes, F. (2001).
\newblock Likelihood inference for {G}ibbs processes in the analysis of spatial
  point patterns.
\newblock {\em International Statistical Review / Revue Internationale de
  Statistique}, 69(1):pp. 81--104.

\bibitem[M{\o}ller et~al., 2010]{moller10}
M{\o}ller, J., Huber, M.~L., and Wolpert, R.~L. (2010).
\newblock Perfect simulation and moment properties for the {M}at{\'{e}}rn
  type-{III} process.
\newblock {\em Stochastic Processes and their Applications},
  120(11):2142--2158.

\bibitem[M{\o}ller and Waagepetersen, 2007]{Moller2007}
M{\o}ller, J. and Waagepetersen, R.~P. (2007).
\newblock {Modern Statistics for Spatial Point Processes}.
\newblock {\em Scandinavian Journal of Statistics}, 34(4):643--684.

\bibitem[Murray and Adams, 2010]{murray2010b}
Murray, I. and Adams, R.~P. (2010).
\newblock Slice sampling covariance hyperparameters of latent {G}aussian
  models.
\newblock In {\em Advances in Neural Information Processing Systems 23}.

\bibitem[Murray et~al., 2010]{murray2010}
Murray, I., Adams, R.~P., and MacKay, D.~J. (2010).
\newblock Elliptical slice sampling.
\newblock {\em JMLR: W\&CP}, 9.

\bibitem[{Peebles}, 1974]{Peebles74}
{Peebles}, P.~J.~E. (1974).
\newblock {The Nature of the Distribution of Galaxies}.
\newblock {\em Astronomy and Astrophysics}, 32:197.

\bibitem[Plummer et~al., 2006]{Rcoda2006}
Plummer, M., Best, N., Cowles, K., and Vines, K. (2006).
\newblock {CODA}: Convergence diagnosis and output analysis for {MCMC}.
\newblock {\em R News}, 6(1):7--11.

\bibitem[Ripley, 1988]{Ripley88}
Ripley, B.~D. (1988).
\newblock {\em Statistical inference for spatial processes / B.D. Ripley}.
\newblock Cambridge University Press, Cambridge [England].

\bibitem[{\c{S}}ahin and S{\"u}ral, 2007]{Sahin2007}
{\c{S}}ahin, G. and S{\"u}ral, H. (2007).
\newblock A review of hierarchical facility location models.
\newblock {\em Computers \& Operations Research}, 34(8):2310 -- 2331.

\bibitem[Scardicchio et~al., 2009]{Scardicchio09}
Scardicchio, A., Zachary, C., and Torquato, S. (2009).
\newblock Statistical properties of determinantal point processes in
  high-dimensional {E}uclidean spaces.
\newblock {\em Phys. Rev. E"}, 79(4):Article 041108.

\bibitem[Strand, 1972]{Strand72}
Strand, L. (1972).
\newblock A model for strand growth.
\newblock In {\em {IUFRO} Third Conference Advisory Group of Forest
  Statisticians, INRA, Institut National de la Recherche Agronomique, Paris.},
  pages 207--216.

\bibitem[van Lieshout and Baddeley, 1996]{vanLies96}
van Lieshout, M. N.~M. and Baddeley, A.~J. (1996).
\newblock A nonparametric measure of spatial interaction in point patterns.
\newblock {\em Statistica Neerlandica}, 50(3):344--361.

\bibitem[Waller et~al., 2011]{WallSar11}
Waller, L.~A., S{\"a}rkk{\"a}, A., Olsbo, V., Myllym{\"a}ki, M.,
  Panoutsopoulou, I.~G., Kennedy, W.~R., and Wendelschafer-Crabb, G. (2011).
\newblock Second-order spatial analysis of epidermal nerve fibers.
\newblock {\em Statistics in Medicine}, 30(23):2827--2841.

\end{thebibliography}

\end{document}